\documentclass[final,3p,times]{elsarticle}

\usepackage{amssymb}
\usepackage{amsmath}
\usepackage{amsthm}

\usepackage{lineno}
\usepackage{todonotes}
\usepackage[most]{tcolorbox}
\usepackage{newtxtext,newtxmath}

\usepackage{thm-restate}

\newtheorem{theorem}{Theorem}
\newtheorem{lemma}{Lemma}
\newtheorem{observation}{Observation}
\newtheorem{definition}{Definition}
\newtheorem{corollary}{Corollary}

\usepackage{xcolor}
\usepackage[normalem]{ulem} 

\begin{document}

\begin{frontmatter}



\title{On the complexity of constrained reconfiguration and motion planning}


\author[udem,ulyon]{Nicolas Bousquet}
\author[uw]{Remy El Sabeh}
\author[aub]{Amer E. Mouawad}
\author[uw]{Naomi Nishimura}

\affiliation[udem]{organization={CNRS - Université de Montréal CRM - CNRS},
            country={Canada}}
            
\affiliation[ulyon]{organization={Univ. Lyon, Université Lyon 1, CNRS, LIRIS UMR CNRS 5205},
            postcode={F-69621}, 
            state={Lyon},
            country={France}}

\affiliation[uw]{organization={David R. Cheriton School of Computer Science, University of Waterloo},
            city={Waterloo},
            state={ON},
            country={Canada}}

\affiliation[aub]{organization={American University of Beirut},
            city={Beirut},
            country={Lebanon}}

\begin{abstract}
Coordinating the motion of multiple agents in constrained environments is a fundamental challenge in robotics, motion planning, and scheduling. A motivating example involves $n$ robotic arms, each represented as a line segment. The objective is to rotate each arm to its vertical orientation, one at a time (clockwise or counterclockwise), without collisions nor rotating any arm more than once. This scenario is an example of the more general $k$-\textsc{Compatible Ordering} problem, where $n$ agents, each capable of $k$ state-changing actions, must transition to specific target states under constraints encoded as a set $\mathcal{G}$ of $k$ pairs of directed graphs.
We show that $k$-\textsc{Compatible Ordering} is $\mathsf{NP}$-complete, even when $\mathcal{G}$ is planar, degenerate, or acyclic. On the positive side, we provide polynomial-time algorithms for cases such as when $k = 1$ or $\mathcal{G}$ has bounded treewidth. We also introduce generalized variants supporting multiple state-changing actions per agent, broadening the applicability of our framework. These results extend to a wide range of scheduling, reconfiguration, and motion planning applications in constrained environments.    
\end{abstract}




\end{frontmatter}




\makeatletter
\def\ps@pprintTitle{%
   \let\@oddhead\@empty
   \let\@evenhead\@empty
   \let\@oddfoot\@empty
   \let\@evenfoot\@oddfoot}
\makeatother

\section{Introduction}
\label{sec:introduction}
Motion planning of multiple robots sharing a common workspace is a well-studied problem that has been tackled from both centralized~\cite{HueA22} and decentralized~\cite{CurkovicJ10,SinaFB18} perspectives. The main goal of robot motion planning is typically to coordinate the actions of robots, such that the robots attain a particular common goal, while guaranteeing that the interference between the robots is either minimized or eliminated. In other words, the movements of a robot should not cause a collision with any of the other robots, which may be moving either synchronously or asynchronously. 

The system can abstractly be defined by a starting position for each robot, a set of potential moves for each robot, which may be dependent on previous moves, some measure we would like to minimize (such as time) and a target position (or positions) for each robot. Many state-of-the-art algorithms for motion planning in the literature are based on sampling from reasonable local choices. While these algorithms may perform well in practice, they are not deterministic and can usually only be assessed empirically against each other, without concrete guarantees. In this paper, and inspired by the abstract definition of the motion planning problem, as well as its accompanying geometric constraints, we tackle motion planning from the reconfiguration angle.

The reconfiguration framework studies the transformation of combinatorial objects via a sequence of operations, each involving a single object, while maintaining feasibility throughout the transformation. The objects under consideration typically belong to the solution space of a well-defined problem. Given source and target configurations, the goal is to determine whether there exists a sequence of valid steps to transform the former into the latter.

A well-known example is the Rubik's cube, where one transforms a source (mixed) configuration to a target (monochromatic) one through a sequence of rotations. Reconfiguration problems can often be modeled as reachability questions in a graph of configurations, where nodes represent feasible solutions, and edges correspond to possible transformations. However, the solution space is usually exponential, so a polynomial-time solution to such problems cannot rely on exhaustive search.

The framework of reconfiguration was formally introduced by Hearn and Demaine with the \textsc{Nondeterministic Constraint Logic} problem, which established $\mathsf{PSPACE}$-completeness results for many naturally occurring games~\cite{hearn:games,hearn:sliding,buchin:dots}. Since then, reconfiguration has extended to graph problems~\cite{ito:reconf}, such as independent set, vertex cover, and dominating set reconfiguration; in the context of these problems, and more broadly, graph reconfiguration problems, configurations are transformed by adding, removing, or substituting vertices.

While some reconfiguration problems admit polynomial-time algorithms on restricted graph classes (e.g., planar graphs~\cite{bousquet:feedback}, graphs of bounded treewidth~\cite{mouawad:parameterized}, or interval graphs~\cite{brianski:interval}), many remain $\mathsf{PSPACE}$-complete or hard to approximate in general~\cite{demaine:tetris,ohsaka:gap2022,ohsaka:gap2024}. For more on reconfiguration problems, we refer the reader to surveys on the topic~\cite{BousquetMNS22+,Heuvel13,Nishimura17}.

Going back to motion planning, we now present variants of problems on robotic arms under the umbrella of reconfiguration problems. Most previous results consider a small number of robots using a broad set of possible moves. We have chosen to restrict the power of robots to study the impact of the number and relative positions of robots on the complexity of the problem. We also focus on centralized algorithms, avoiding synchronization issues that might arise in multi-agent systems. We refer to the most general setup as the \textsc{Robotic Arm Motion Planning (RAMP)} problem.

The presentation of the \textsc{RAMP} problem, followed by its variants, is meant to provide context for and a motivation of the paper's two main problems, namely the $k$-\textsc{Compatible Ordering} problem and the $k$-\textsc{Compatible Set Arrangement} problem, with the included results for these two problems being used to refine what we know about the tractability of different versions of the \textsc{RAMP} problem. \textsc{RAMP} problems will therefore be recurrently used as tangible examples of $k$-compatible ordering and arrangement problems. 

The \textsc{RAMP} problem is defined in the plane $\mathbb{R}^2$, with the canonical $x$-axis and $y$-axis, where a \emph{robotic arm} is a unit-length line segment. An arm can be uniquely identified by its \emph{center}, i.e., the midpoint $(x, y)$ of its line segment, and the \emph{angle} $a \in [0, \pi)$ in radians that the line segment forms with the horizontal line passing through its center. A collection of arms is \emph{conflict-free} whenever the corresponding line segments are pairwise non-intersecting. We are particularly interested in arms that are \emph{horizontal}, which have an angle that is equal to $0$, or \emph{vertical}, which have an angle that is equal to $\pi/2$. A reconfiguration operation on a set of arms is a \emph{rotation} of a single arm about its center, which corresponds to the robotic arm in question updating its angle $a$ to $a' = (a + \alpha) \text{ mod } \pi$, where $\alpha \neq 0$. A rotation is said to be a \emph{counterclockwise rotation} when $\alpha > 0$, and is said to be a \emph{clockwise rotation} otherwise. When a robotic arm rotates from one orientation to another, it sweeps two opposite sectors at the center of the arm, associated with the arcs defined between $a$ and $a'$. If a robotic arm $r_1$ is rotating, and another arm $r_2$ intersects one of the two sectors swept by the rotation of $r_1$, we say that the rotation of $r_1$ \emph{hits} $r_2$ (see Figure~\ref{fig:toothpicks} for an illustration).

We now formally define the problems under consideration, starting with the most general formulation of the \textsc{RAMP} problem:
\smallskip

\begin{tcolorbox}[colback=white, colframe=black,  
                  arc=4mm, boxrule=0.3mm, 
                  fonttitle=\bfseries]
\textsc{Multiple Move Robotic Arm Motion Planning (MM-RAMP)}
\\
\textbf{Input: } A sequence of robotic arms $\mathcal{R} = r_1, \ldots, r_n$ with distinct fixed centers, two conflict-free configurations of angles $\mathcal{R}_{\text{s}}= \alpha_1,\ldots,\alpha_n$ and $\mathcal{R}_{\text{t}}=\beta_1,\ldots,\beta_n$, and for $1\leq i \leq n$, a collection of angles $\mathcal{A}_i$ such that both $\alpha_i$ and $\beta_i$ belong to $\mathcal{A}_i$.
\\
\textbf{Output:} Whether there exists a sequence of conflict-free robotic arm configurations $\mathcal{R}_1:= \mathcal{R}_{\text{s}}, \ldots, \mathcal{R}_p:= \mathcal{R}_{\text{t}}$, such that, for every pair $(\mathcal{R}_j, \mathcal{R}_{j + 1})$ of consecutive configurations, $1 \leq j \leq p - 1$, there exists a robotic arm $r_i$ satisfying the following:
\begin{itemize}
\item $\mathcal{R}_j$ and $\mathcal{R}_{j + 1}$ differ only in the angle of $r_i$ (i.e., consecutive configurations differ in the rotation of a single arm);
\item $r_i$ is rotated from $a \in \mathcal{A}_i$ in $\mathcal{R}_{j}$ to $a' \in \mathcal{A}_i$ in $\mathcal{R}_{j + 1}$; and
\item The rotation of $r_i$ does not hit any other robotic arm in $\mathcal{R}_j$.
\end{itemize}
\end{tcolorbox}

We also consider a simpler problem, interesting in its own right, that captures much of the difficulty of the general problem. In particular, instead of allowing multiple moves per arm, we seek a transformation where each arm moves at most once, and where the final desired orientation for all arms is the vertical one. More formally, we define the following problem: 
\smallskip

\begin{tcolorbox}[colback=white, colframe=black,  
                  arc=4mm, boxrule=0.3mm, 
                  fonttitle=\bfseries]
\textsc{Single Move Robotic Arm Motion Planning (SM-RAMP)}
\\
\textbf{Input: } A sequence of robotic arms $\mathcal{R} = r_1, \ldots, r_n$, with distinct fixed centers.
\\
\textbf{Output:} Whether there exists a sequence of conflict-free robotic arm configurations $\mathcal{R}_0, \ldots, \mathcal{R}_n$ such that $\mathcal{R} = \mathcal{R}_0$ and all robotic arms are vertical in $\mathcal{R}_n$. Moreover, for every pair $(\mathcal{R}_j, \mathcal{R}_{j + 1})$ of consecutive configurations, $0 \leq j \leq n - 1$, there exists a robotic arm $r_i$ satisfying the following:
\begin{itemize}
\item $\mathcal{R}_j$ and $\mathcal{R}_{j + 1}$ differ only in the angle of $r_i$;
\item $r_i$ is rotated from $a \in \mathcal{A}_i$ in $\mathcal{R}_{j}$ to $\pi/2$ (vertical orientation) in $\mathcal{R}_{j + 1}$;
\item The rotation of $r_i$ does not hit any other robotic arm in $\mathcal{R}_j$; and
\item $r_i$ is rotated at most once in the sequence $\mathcal{R}_0, \ldots, \mathcal{R}_n$. 
\end{itemize}
\end{tcolorbox}

In order to utilize what is known in reconfiguration on graphs in terms of tractability and how it relates to a graph's structure or its properties, it is reasonable to attempt to define both \textsc{MM-RAMP} and \textsc{SM-RAMP} in graph theoretic terms. Therefore, we consider two problems on graphs that generalize them, starting with the $k$-\textsc{Compatible Ordering} problem. The $k$-\textsc{Compatible Ordering} problem, which is a generalization of \textsc{SM-RAMP}, is defined below:
\smallskip

\begin{tcolorbox}[colback=white, colframe=black,  
                  arc=4mm, boxrule=0.3mm, 
                  fonttitle=\bfseries]
$k$-\textsc{Compatible Ordering}
\\
\textbf{Input: } A vertex set $V$, a collection $\mathcal{G}$ of $k$ pairs of directed graphs $(A_1, B_1), \ldots, (A_k, B_k)$, such that $V(A_i) = V(B_i) = V$ for $1 \leq i \leq k$. 
~\\
\textbf{Output: } Whether there exists a pair $(\mathcal{S}, \mathcal{L})$, where $\mathcal{S} = s_1, \ldots, s_n$ is an ordering of $V$,  $\mathcal{L} = \ell_1, \ldots, \ell_n$ is a sequence of labels, and vertex $s_i \in \mathcal{S}$ is assigned label $\ell_i \in [k]$, for $1 \leq i \leq n$, such that the following two constraints are satisfied:
\begin{itemize}
\item $s_i$ is a sink in $A_{\ell_i}[s_i, \ldots, s_n]$; and
\item $s_i$ is a source in $B_{\ell_i}[s_1, \ldots, s_i]$.
\end{itemize}
\end{tcolorbox}

\begin{observation}
    There exists a polynomial-time reduction from the \textsc{SM-RAMP} problem to the $2$-\textsc{Compatible Ordering} problem.
\end{observation}

\begin{proof}
Consider an instance $\mathcal{R} = r_1, \ldots, r_n$ of \textsc{SM-RAMP}. For each robotic arm $r_i \in \mathcal{R}$, we create a vertex $v_i$ in $V$. If the clockwise (resp. counterclockwise) rotation of an arm $r_i$ hits an arm $r_j$, we add the arc $r_ir_j$ to $A_1$ (resp. $A_2$). If the vertical orientation of an arm $r_i$ is hit by the clockwise (resp. counterclockwise) rotation of an arm $r_j$, or crosses $r_j$ in its initial orientation, we add the arc $r_ir_j$ to $B_1$ (resp. $B_2$). The vertex set $V$ and the four directed graphs $A_1$, $B_1$, $A_2$, $B_2$ can clearly be constructed in time polynomial in $n$, the number of robotic arms in $\mathcal{R}$. We claim that $\mathcal{R}$ is a yes-instance of \textsc{SM-RAMP} if and only if $\{(V, \mathcal{G} = \{(A_1, B_1), (A_2, B_2)\}$ is a yes-instance of $2$-\textsc{Compatible Ordering}. We only present the forward direction of the proof, as the argument for the backward direction is symmetrical and follows directly from the construction of the $2$-\textsc{Compatible Ordering} instance.

Let $(\mathcal{R}_{\text{sol}}, \mathcal{D})$ represent a solution to the \textsc{SM-RAMP} instance. More specifically, let $\mathcal{R}_{\text{sol}} = r'_1, \ldots, r'_n$ be an ordering of the robotic arms in $\mathcal{R}$ such that $r'_i$ is the $i$th robotic arm to be rotated, for $1 \leq i \leq n$, and let $\mathcal{D} = d_1, \ldots, d_n$ denote the direction in which the robotic arms are rotated, such that $d_i \in \{r, l\}$, for $1 \leq i \leq n$ (where $r$ and $l$ denote a clockwise and a counterclockwise rotation, respectively). We claim that $(\mathcal{S} = v'_1, \ldots, v'_n, \mathcal{L} = \ell_1, \ldots, \ell_n)$, where $v'_i$ is the vertex associated with $r'_i$ in the construction, and where $\ell_i = 1$ if $d_i = r$ and $\ell_i = 2$ if $d_i = l$, is a valid labeled ordering for the constructed \textsc{SM-RAMP} instance.
Assume otherwise, then, there must exist two vertices $v'_i$ and $v'_j$ in $\mathcal{S}$ (without loss of generality, $i < j$) that are assigned labels $\ell_i$ and $\ell_j$ such that either $v'_iv'_j$ is an arc in $A_{\ell_i}$, or $v'_iv'_j$ is an arc in $B_{\ell_j}$. However, by construction, the existence of the former of these two arcs implies that the rotation of $r'_i$ in direction $d_i$ would have hit $r'_j$ in its initial orientation, and that of the latter of these two arcs implies that the rotation of $r'_j$ in direction $d_j$ would have hit $r'_i$ in its vertical orientation, or that $r'_i$ in its vertical orientation crosses $r'_j$ in its initial orientation: either of these two cases invalidates $(\mathcal{R}_{\text{sol}}, \mathcal{D})$ being a solution to the \textsc{SM-RAMP} instance.
\end{proof}

While the \textsc{SM-RAMP} problem, which we initially set to solve, inspired the $k$-\textsc{Compatible Ordering} problem, the latter's applications extend beyond problems on robotic arms, as it can capture constraints that are not inherently geometric. More precisely, the $k$-\textsc{Compatible Ordering} problem models a class of element-ordering problems in which each element admits one of $k$ possible insertion modes, where the chosen mode determines the set of elements that must precede it and the set that must follow it, independently of the insertion modes of these other elements.

We use the \textsc{Truck Unloading} problem as an example of a $3$-\textsc{Compatible Ordering} problem. Let $b_1, \ldots, b_n$ be a set of boxes in a truck, where each box is a rectangular cuboid. We are looking to place these boxes in a single stack on a platform inside a warehouse, while ensuring that the stack does not topple. Each box may be placed in the stack with any of its six faces oriented downward; we will assume that opposite faces count as the same orientation, so each box has three orientations, each of which corresponds to a distinct face oriented downward (where no two of those distinct faces are opposite to each other). The three faces are numbered $1$ to $3$ arbitrarily. The orientation of a box determines its weight distribution, which should be taken into account when it comes to balancing the stack. Choosing label $\ell$ for box $b_i$ means: ``put box $b_i$ in the stack with face $\ell$ facing downward''. Getting a box onto the stack on the platform is a two-step process: the first step involves getting the box out of the truck, and the second step involves placing the box at the top of the stack, in such a way that does not cause the stack to topple. We add one more restriction: the orientation of the box has to be chosen before it is lifted (i.e., before the first step), and the box preserves its orientation until it is placed on top of the stack. $A$ graph constraints are associated with the first step, whereas $B$ graph constraints are associated with the second step. More specifically, the graph $A_\ell$ has an arc $b_ib_j$ whenever moving $b_i$ outside of the truck in orientation $\ell$ must wait until $b_j$ has already been moved onto the stack (e.g., $b_i$ sits deeper in the truck than $b_j$, and cannot be carried out of the truck in orientation $\ell$ without hitting $b_j$). The graph $B_\ell$ has an arc $b_ib_j$ whenever $b_i$ being added to the stack in any orientation before $b_j$ in orientation $\ell$ would cause the stack to topple, due to the weight distribution in $b_i$ prohibiting $b_j$ from being placed on top of it in orientation $\ell$, irrespective of $b_i$'s orientation. Thus, when a box $b_i$ is assigned label $\ell$ in a solution's ordering, we simultaneously enforce that all its $A_\ell$-successors appear earlier than it in the ordering, ensuring that no box prevents box $b_i$ from being moved out of the truck, and all its $B_\ell$-predecessors appear later than it in the ordering of $V$ in a solution, ensuring that no box positioned under $b_i$ in the stack causes the stack to topple.

While both the \textsc{SM-RAMP} problem and the \textsc{Truck Unloading} problem can be modeled as $k$-\textsc{Compatible Ordering} problems ($k = 2$ for \textsc{SM-RAMP}, $k = 3$ for \textsc{Truck Unloading}), the value of $k$ is fixed and is intrinsic to the definition of both problems. With that being said, \textsc{Truck Unloading} can be generalized more naturally by replacing boxes with polyhedrons that can be oriented in one of $k$ ways. We have therefore decided to include the parameter $k$ in the formulation of $k$-\textsc{Compatible Ordering} not only to capture a wider array of problems, but also to analyze how the value of $k$ affects the problem's tractability.

\begin{figure}
    \centering
    \includegraphics[width=0.75\linewidth]{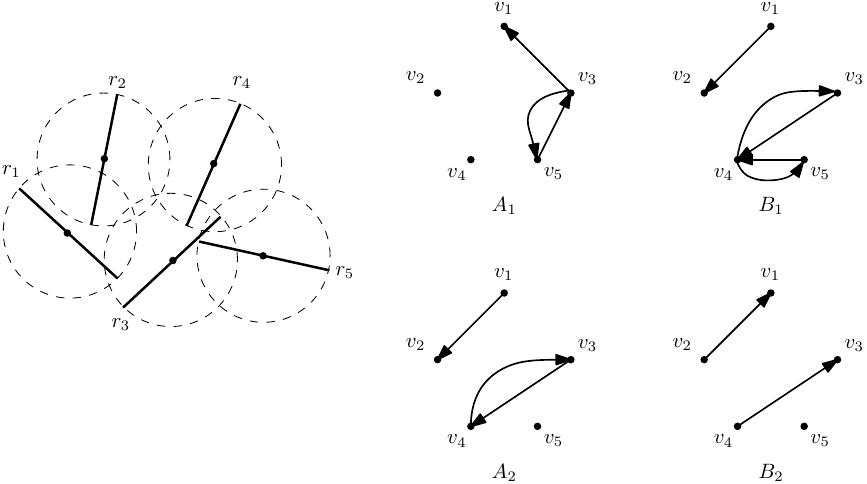}
    \caption{An \textsc{SM-RAMP} instance (left) and its equivalent $2$-\textsc{Compatible Ordering} instance (right). Robotic arms can move towards the vertical orientation in the order $r_1,r_5,r_3,r_4,r_2$, where $r_1,r_3$ rotate clockwise and $r_2,r_4,r_5$ rotate counterclockwise. The arc $v_3v_5$ in $A_1$ means that a clockwise rotation of $r_3$ to its final orientation would hit $r_5$ if $r_5$ has not rotated. The absence of the arc $v_3v_5$ from $B_2$ means that $r_3$ does not block the counterclockwise rotation of $r_5$ if $r_3$ is already vertical. The solution of the $2$-\textsc{Compatible Ordering} instance corresponding to the solution of \textsc{SM-RAMP} described above is the solution $\mathcal{S}=v_1,v_5,v_3,v_4,v_2$ and $\mathcal{L}=1,2,1,2,2$.}
    \label{fig:toothpicks}
\end{figure}

\textsc{SM-RAMP} has a counterpart where robots can move back and forth before reaching their final orientations, which is the \textsc{MM-RAMP} problem defined earlier. One then wonders whether we can generalize \textsc{MM-RAMP} as we generalized \textsc{SM-RAMP}. The answer is positive. The following problem is a natural generalization of \textsc{MM-RAMP}:
\smallskip 

\begin{tcolorbox}[colback=white, colframe=black,  
                  arc=4mm, boxrule=0.3mm, 
                  fonttitle=\bfseries]
$k$-\textsc{Compatible Set Arrangement}
\\\
\textbf{Input: } A vertex set $V$, a collection $\mathcal{G}$ of $k$ pairs of directed graphs $(A_1, B_1), \ldots, (A_k, B_k)$, such that $V(A_i) = V(B_i) = V$ for $1 \leq i \leq k$, and two sets $V_{\text{s}}, V_{\text{t}} \subseteq V$.
\\\textbf{Output: } Whether there exists a sequence of sets $\mathbb{V} = V_1, \ldots, V_q$, such that:
\begin{itemize}
    \item $V_{\text{s}} = V_1$ and $V_{\text{t}} = V_q$;
    \item $|V_j \Delta V_{j + 1}| = 1$ for $1 \leq j \leq q - 1$ (i.e., a vertex is either added to or removed from $V_j$ to obtain $V_{j + 1}$); and
    \item For each vertex $v$ that is added to or removed from a set $V_j$, there exists a label $\ell \in [k]$ assigned to $v$ such that:
    \begin{itemize}
        \item If $v$ is added to $V_j$, $v$ is a sink in $A_\ell[V_{j + 1}]$ and a source in $B_\ell[V \setminus \left(V_{j + 1} \setminus \{v\}\right)]$. 
        \item If $v$ is removed from $V_j$, $v$ is a sink in $A_\ell[V_{j}]$ and a source in $B_\ell[V \setminus \left(V_{j} \setminus \{v\}\right)]$.
    \end{itemize}
\end{itemize}   
\end{tcolorbox}

\paragraph{Our contributions}
The main technical contributions of the paper are efficient algorithms to find $k$-compatible orderings, and hardness results for $k$-\textsc{Compatible Ordering} and $k$-\textsc{Compatible Set Arrangement}. 
We first prove the following dichotomy (Section~\ref{sec:positive_results}):

\begin{theorem}
\label{thm:kcompac_dico}
$k$-\textsc{Compatible Ordering} can be decided in polynomial time if $k=1$ and is $\mathsf{NP}$-complete for every $k \ge 2$.
\end{theorem}

The theorem ensures, in particular, that if all the robotic arms are allowed to rotate only in one direction (clockwise for instance), then we can decide the problem in polynomial time. Since we have two directions and since $k$-\textsc{Compatible Ordering} is $\mathsf{NP}$-complete for $k = 2$, if \textsc{SM-RAMP} can be decided in polynomial time, it must be due to some topological or geometric structure of the constraints. 

In that direction, we also give some evidence that only very specific topologies can help.  Namely, we prove that $2$-\textsc{Compatible Ordering} remains $\mathsf{NP}$-complete even when the union of all the graphs in $\mathcal{G}$ is planar, or more formally: 

\begin{restatable}{theorem}{thmNPPlanar}
\label{thm:kcompat-npc}
$k$-\textsc{Compatible Ordering} is $\mathsf{NP}$-complete for $k \geq 2$ even when restricted to instances where the union of all graphs in $\mathcal{G}$ is planar and has degeneracy $3$.
\end{restatable}

The theorem above shows that the topology and having few constraints locally (bounded degeneracy) do not help to obtain polynomial-time algorithms. The proof has a drawback: some of the graphs in $\mathcal{G}$ contain directed cycles. As far as yes-instances of \textsc{SM-RAMP} are concerned, such cycles do not exist, as one can prove that a set of arms cannot prevent themselves from rotating in a ``cyclic'' manner if a solution exists.  One may then wonder what happens when we restrict graphs in $\mathcal{G}$ to be acyclic; we prove that the problem remains $\mathsf{NP}$-complete. However, we prove that if we strengthen this assumption a bit, the problem becomes polynomial-time solvable. Namely, if there exists an integer $x$ such that the union of the pair $(A_x,B_x)$ forms an acyclic graph, then $k$-\textsc{Compatible Ordering} can be decided in polynomial time; unfortunately, in the case of robotic arms, a cycle is possible in the union of the two directed graphs representing constraints on the same rotation direction.

\begin{restatable}{theorem}{thmNPAcyclic}
\label{thm:kcompat-npc-acyclic}
   $k$-\textsc{Compatible Ordering} is $\mathsf{NP}$-complete for $k \geq 2$ even when restricted to instances where none of the graphs in $\mathcal{G}$ contains a directed cycle. 
   \end{restatable}

Next, we investigate what other restrictions make the problem easy to solve. We present several polynomial-time algorithms. Our first result proves that the problem is polynomial-time solvable if the union of all (directed) graphs has bounded (undirected) treewidth. We note that this is not a consequence of the famous Courcelle meta-theorem~\cite{DBLP:journals/eatcs/CourcelleE12} stating that all problems that can be expressed in monadic second-order logic can be decided in polynomial time on bounded treewidth graphs.  In a sense, our result can be interpreted as showing that when robotic arms ``weakly'' interact, we obtain a polynomial-time algorithm.

\begin{restatable}{theorem}{thmTreewidth}
    \label{thm:bdd-treewidth}$k$-\textsc{Compatible Ordering} can be solved in time $\mathcal{O}(|V| \cdot (t^2 \cdot t! \cdot k^t))$, where $t$ is the treewidth of the undirected union of the graphs in $\mathcal{G}$.
\end{restatable}

Our second and final positive result also exploits the expected structure of the $k$-\textsc{Compatible Ordering} instances encoding \textsc{SM-RAMP} instances. Indeed, since robotic arms only have constraints related to nearby robotic arms, one should be able to iteratively encode all the constraints by considering sets of robotic arms, rather than individual arms. This partitioning is such that two robotic arms belonging to the same set have the same relationship with respect to the robotic arms outside of the set, without partitioning the robotic arms into too many sets. This idea is tightly coupled with the notions of modular decomposition and modular width, which we modify slightly to get them to work for $k$-\textsc{Compatible Ordering}, with the existence of labels and directions on the edges. We prove that if the union of the graphs has bounded labeled modular width (see Section~\ref{sec:positive_results} for a formal definition), then the problem can be decided in polynomial time. We note that two robotic arms (or, in the context of $k$-\textsc{Compatible Ordering}, two vertices) belonging to the same set are allowed to have different constraints involving the rest of the robotic arms (or vertices) in the same set. We provide intuition for the modular width using the unloading problem. For instance, if the objects are stored in several boxes in a container, we need to completely empty and remove a box before being able to open another one. Thus, all the objects in the same box behave the same with respect to other boxes.

\begin{restatable}{theorem}{thmModularwidth}
\label{thm:bdd-mdw}
$k$-\textsc{Compatible Ordering} can be solved in time $f(k + mw) \cdot |V|^{\mathcal{O}(1)}$, where $f$ is a computable function and $mw$ is the labeled modular width, assuming a labeled modular decomposition of the labeled union of $\mathcal{G}$ is given as part of the input. Thus, the problem is solvable in polynomial time when $k + mw$ is a constant. 
\end{restatable}

In an attempt to narrow down what makes $k$-\textsc{Compatible Ordering} hard, we consider whether looking for an orderable subset of the vertices, subject to the same constraints, makes the problem any easier to solve. In particular, we investigate whether the problem can be solved in polynomial time when the subset is bounded in size. We prove that this is likely false.

\begin{restatable}{theorem}{thmBoundedOrdering}
    $k$-\textsc{Compatible Bounded Ordering} is $\mathsf{W[1]}$-hard parameterized by $b$.
\end{restatable}

Our final result concerns the $k$-\textsc{Compatible Set Arrangement} problem. Unlike the $k$-\textsc{Compatible Ordering} problem, which allows only one move per agent (or robotic arm), the $k$-\textsc{Compatible Set Arrangement} problem allows multiple moves per agent and is therefore not surprisingly ``harder'' to solve. Our next result confirms this intuition by showing that the problem is indeed $\mathsf{PSPACE}$-complete.

\begin{restatable}{theorem}{thmArrangement}
    \label{thm:arrangement}
    $k$-\textsc{Compatible Set Arrangement} is $\mathsf{PSPACE}$-complete, even when restricted to instances where the union of the graphs in $\mathcal{G}$ is planar, has maximum degree 6, and has bounded bandwidth.
\end{restatable} 

The paper is organized as follows: Section~\ref{sec:preliminaries} introduces the necessary preliminaries. Section~\ref{sec:k_compat_ordering} presents the results concerning the $k$-\textsc{Compatible Ordering} problem, including positive results in Section~\ref{sec:positive_results} and hardness results in Section~\ref{sec:hardness_results}. Section~\ref{sec:bounded_ordering} covers a variant of $k$-\textsc{Compatible Ordering} where we are interested in ordering a subset of the vertices in the input, subject to the usual constraints. Section~\ref{sec:compat_set_arrangement} covers the $\mathsf{PSPACE}$-completeness result for the $k$-\textsc{Compatible Set Arrangement} problem. Section~\ref{sec:conclusion} summarizes our main contributions and suggests potential avenues for future research, including but not limited to problems on robotic arms.

\section{Preliminaries}
\label{sec:preliminaries}
For $k \in \mathbb{N}$, we denote the set $\{1, \ldots, k\}$ by $[k]$. Given a graph $G$, we use $V(G)$ and $E(G)$ to refer to the vertex set and edge set of $G$, and we use $n$ and $m$ to refer to $|V(G)|$ and $|E(G)|$ respectively.
In a directed graph, a \emph{cycle} is a sequence of vertices $v_1, \ldots, v_n$ such that $v_iv_{i + 1}$ is an arc for $1 \leq i \leq n - 1$, $v_n$ is adjacent to $v_1$, and no vertices are repeated. A vertex $v$ is said to be an \emph{out-neighbor} (resp. an \emph{in-neighbor}) of $u$ in a directed graph $G$ if the arc $uv$ (resp. $vu$) belongs to the edge set of $G$.  The \emph{indegree} (resp. the \emph{outdegree}) of a vertex $v \in V(G)$ is the number of its in-neighbors (resp. out-neighbors). A vertex $v$ is said to be a \emph{sink} (resp. a \emph{source}) in a directed graph if it has no out-neighbors (resp. in-neighbors). A directed graph $G$ is said to be a \emph{directed acyclic graph} (DAG) if it does not contain any cycles. In a directed acyclic graph on $n$ vertices, a sequence of its vertices $v_1, \ldots, v_n$ is a \emph{topological ordering} if there does not exist an arc $v_iv_j$ in the graph, where $v_j$ occurs earlier than $v_i$ in the list. If a directed graph is edge-weighted, the \emph{in-weight} (resp. \emph{out-weight}) of a vertex $v$ is the sum of the weights of the edges going into (resp. going out of)~$v$. 

Given a graph $G$ and a subset $X \subseteq V(G)$, we denote the induced subgraph of $G$ on the vertices in $X$ by $G[X]$. Let the \emph{union} of a set of graphs that have the same vertex set $G_1 = (V, E_1), G_2 = (V, E_2), \ldots, G_k = (V, E_k)$ be a graph $G' = (V, \{ v_iv_j$ $|$ $v_iv_j \in \bigcup\limits_{i = 1}^{k} E_i\})$. If $G_1, \ldots, G_k$ are directed graphs, and we intend $G'$ to be undirected, we call the operation \emph{undirected union} instead.  For brevity, whenever we refer to the union of $\mathcal{G} = \{(A_1, B_1), \ldots, (A_k, B_k)\}$, we are referring to the union of the graphs in $\mathcal{G}$.

A graph is said to be \emph{planar} if its edges can be drawn in a plane as simple curves such that the drawings of two edges share an endpoint if and only if the edges themselves share an endpoint.

A $d$\emph{-degenerate} graph is an undirected graph in which every induced subgraph has a vertex of degree at most $d$; the \emph{degeneracy} of a graph is the smallest value of $d$ for which the graph is $d$-degenerate. The \emph{treewidth} of an undirected graph informally measures how close the graph is to a tree. We direct the reader to the relevant work for terminology and definitions related to the notions of treewidth and tree decomposition~\cite{cygan:parameterized}. Intuitively, a \emph{modular decomposition} of an undirected graph is a recursive partitioning of its vertex set, such that the neighbors of the vertices in the same partition found in the other partitions at the same recursion level are similar. The \emph{modular width} of an undirected graph is defined with respect to its modular decompositions, and is equal to the smallest maximum partition size taken across all its possible modular decompositions. Those two concepts are presented in detail in Section~\ref{sec:positive_results}. The \emph{bandwidth} of an undirected graph $(V, E)$ is a measure of how tightly the vertices of a graph can be arranged in a linear order. It is given by $\min\limits_{f} \max\{|f(v_i) - f(v_j)| : v_iv_j \in E\}$, where $f : V \mapsto \mathbb{N}$ is a one-to-one function. The degeneracy, treewidth, and bandwidth of a directed graph all refer to the respective measures of the underlying undirected graph.

\section{Finding $k$-compatible orderings}
\label{sec:k_compat_ordering}

In this section, we present the positive (Section~\ref{sec:positive_results}) and the hardness results (Section~\ref{sec:hardness_results}) for $k$-\textsc{Compatible Ordering}. 
\\
\\
\indent Recall that in the definition of $k$-\textsc{Compatible Ordering}, we are given a collection of constraints $\mathcal{G}$ and we seek a pair $(\mathcal{S}, \mathcal{L})$, where $\mathcal{S} = s_1, \ldots, s_n$ is an ordering of $V$,  $\mathcal{L} = \ell_1, \ldots, \ell_n$ is a sequence of labels, and vertex $s_i \in \mathcal{S}$ is assigned label $\ell_i \in [k]$ such that $s_i$ satisfies the following two constraints:
\begin{enumerate}[(i)]
\item $s_i$ is a sink in $A_{\ell_i}[s_i, \ldots, s_n]$; and
\item $s_i$ is a source in $B_{\ell_i}[s_1, \ldots, s_i]$.
\end{enumerate}

We present a set of definitions related to the problem statement.

\begin{definition}
    A directed graph $A_i$ in $\mathcal{G}$ is called an \textbf{$A$ graph}, whereas a directed graph $B_j$ in $\mathcal{G}$ is called a \textbf{$B$ graph}.
\end{definition}

\begin{definition}
    Given a solution $(\mathcal{S}, \mathcal{L})$ to a $k$-\textsc{Compatible Ordering} instance, we say that $s_i \in \mathcal{S}$ satisfies both \textbf{the sink constraint} (constraint (i)) and \textbf{the source constraint} (constraint (ii)).
\end{definition}

\begin{definition}
    Given a solution $(\mathcal{S}, \mathcal{L})$ to a $k$-\textsc{Compatible Ordering} instance, we refer to the sequence of labels $\mathcal{L}$ as a \textbf{labeling}, and we call $(\mathcal{S}, \mathcal{L})$ a \textbf{valid labeled ordering} (of the vertices in the instance), as opposed to an arbitrary labeled ordering, where the source and sink constraints need not be satisfied.
\end{definition}

\begin{figure}
    \centering
    \includegraphics[width=0.5\linewidth]{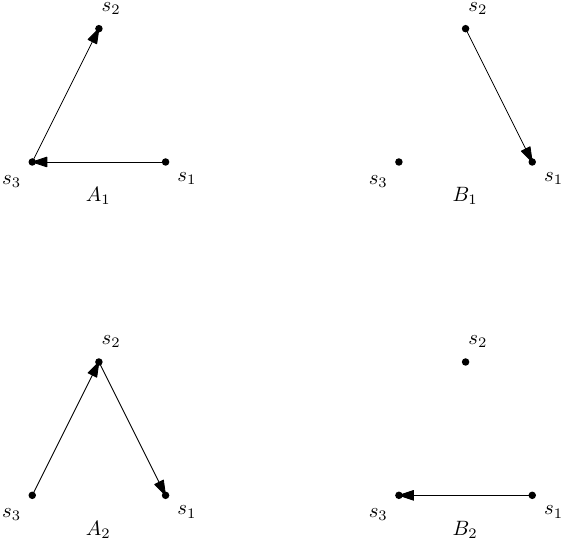}
    \caption{Example where taking the union of the $A$ graph and the $B$ graph of the same label turns a yes-instance of $k$-\textsc{Compatible Ordering} into a no-instance.}
    \label{fig:AB-merge}
\end{figure}

At times, in particular for the positive results in this section, we will find it useful to consider a $k$-\textsc{Compatible Ordering} instance derived from another instance by considering a subset of the vertex set in the input. 

\begin{definition}
    Given an instance $\mathcal{I} = (V, \{(A_1, B_1), \ldots, (A_k, B_k)\})$ of $k$-\textsc{Compatible Ordering}, an \textbf{induced instance} of $\mathcal{I}$ on $V' \subseteq V$ is the instance $\mathcal{I}' = (V', \{(A_1[V'], B_1[V']), \ldots, (A_k[V'], B_k[V'])\})$.
\end{definition}

The source and sink constraints lead to the following characterizations of any solution to the $k$-\textsc{Compatible Ordering} problem, for any label $\ell$ and vertices $v_i$ and $v_j$.

\begin{observation}
    \label{obs:ordering}If $v_iv_j$ is an arc in $A_\ell$ (resp. $B_\ell$), and $v_i$ (resp. $v_j$) is assigned label $\ell$, then $v_j$ must precede $v_i$ in the ordering.
\end{observation}
\begin{observation}
    \label{obs:before}If $v_iv_j$ is an arc in all $A$ graphs or all $B$ graphs, then $v_j$ must precede $v_i$ in the ordering. 
\end{observation}
\begin{observation}
    \label{obs:directedcycle}For a set of vertices $V' \subseteq V$ forming a directed cycle in a graph $A_\ell$ or $B_\ell$, the vertices in $V'$ cannot all be labeled $\ell$.
\end{observation}

\begin{observation}
    \label{obs:forbiddenlabel}If $A_\ell$ contains the arc $v_iv_j$, and $B_\ell$ contains the arc $v_jv_i$, then $v_i$ cannot be labeled $\ell$.
\end{observation}



Observations~\ref{obs:ordering},~\ref{obs:before},~\ref{obs:directedcycle} and~\ref{obs:forbiddenlabel} greatly simplify the presentation of the hardness results in Section~\ref{sec:hardness_results}, as most of the gadgets that are used in the constructions depend on one of, or a combination of, these observations. For now, we have all the definitions needed to proceed to the positive results.

\subsection{Positive Results}\label{sec:positive_results}

In the first stage, and before looking at $k$-\textsc{Compatible Ordering} results on instances with bounded structural parameters, we situate the $2$-\textsc{Compatible Ordering} instances encoding \textsc{SM-RAMP} instances with respect to the rest of the $k$-\textsc{Compatible Ordering} instances. In particular, we look at instances with a single directed graphs per label and those with a single label. 

There are two directed graphs per label in the specification of the input to $k$-\textsc{Compatible Ordering}: one can wonder whether replacing each pair of directed graphs with the union of its graphs yields an equivalent instance (while maintaining the source and sink constraints on the resulting graphs). The answer is negative: Figure~\ref{fig:AB-merge} justifies the use of two directed graphs per label in the input specification of $k$-\textsc{Compatible Ordering}. If we replace each graph pair with the union of the two graphs, and we treat all constraints as sink constraints (or, symmetrically, source constraints), the resulting instance has no valid labeled ordering, whereas a solution to the original instance is $\mathcal{S} = s_1, s_2, s_3$ and $\mathcal{L} = 2, 2, 1$. In fact, substituting the $A$ graph and the $B$ graph of every pair with their union in the example of Figure~\ref{fig:AB-merge} does not preserve solutions specifically because the union is a cycle. If this is not the case, we are dealing with a yes-instance of $k$-\textsc{Compatible Ordering}, and a valid labeled ordering can be computed efficiently (Corollary~\ref{cor:trivial-1-compat-ordering}). We first show how checking for cycles in the union of the graphs is sufficient to determine whether a $1$-\textsc{Compatible Ordering} instance is a yes-instance. 

\begin{lemma}
    \label{lem:new-1-compat-ordering}
    $1$-\textsc{Compatible Ordering} can be solved in time $\mathcal{O}(n + m)$, where $m$ is the number of edges in the union graph.
\end{lemma}

\begin{proof}
    We prove that an instance $(V, \{(A_1, B_1)\})$ of \textsc{1-Compatible Ordering} is a yes-instance if and only if the union of $A_1$ and $B_1$ is a directed acyclic graph.

   For the forward direction, towards a contradiction, suppose that $(V, \{(A_1, B_1)\})$ is a yes-instance, and that the union of $A_1$ and $B_1$ contains at least one cycle.
Consider an arbitrary solution $(\mathcal{S} = s_1, \ldots, s_n, \mathcal{L})$, where $\mathcal{L}$ trivially assigns label $1$ to all vertices. Since the union of $A_1$ and $B_1$ contains at least one cycle, there will exist a vertex $s_x$ in $\mathcal{S}$ that has an out-neighbor $s_w$ in $s_{x + 1}, \ldots, s_n$, otherwise, the union of $A_1$ and $B_1$ would have a topological ordering and would therefore be acyclic, a contradiction. Whether the arc $s_xs_w$ is in $A_1$ or $B_1$, it forces the solution to include $s_w$ before $s_x$, which contradicts $(\mathcal{S}, \mathcal{L})$ being a solution to the instance. 

As for the backward direction, let the labeling $\mathcal{L}$ be the labeling that assigns label $1$ to all vertices. We claim that we can find an ordering $\mathcal{S}$ of the vertices such that $(\mathcal{S}, \mathcal{L})$ is a valid labeled ordering.  We first compute a topological ordering for the union of $A_1$ and $B_1$ in time $\mathcal{O}(n + m)$ ($m$ being the number of edges in the union of $A_1$ and $B_1$). Let $\mathcal{S} = s_1, \ldots, s_n$ be the topological ordering in reverse order.
   To show that $(\mathcal{S}, \mathcal{L})$ is a solution, suppose, for the sake of contradiction, that there exists a vertex $s_x$ in $\mathcal{S}$ that is assigned label $1$ and that violates either the sink or the source constraint. If $s_x$ is not a sink in $A_{1}[s_x, \ldots, s_n]$, then there exists an out-neighbor of $s_x$ in $s_{x + 1}, \ldots, s_n$. However, all the out-neighbors of $s_x$ in $A_i$ are in $s_1, \ldots, s_{x - 1}$, given that $\mathcal{S}$ is a reversed topological ordering. An analogous argument ensures that $s_x$ is a source in $B_{1}[s_1, \ldots, s_x]$. Since the choice of $s_x$ is arbitrary, there is no vertex for which a source or sink constraint is violated, and $(\mathcal{S}, \mathcal{L})$ is a valid labeled ordering, as needed.
\end{proof}

A generalization of the backward direction of Lemma~\ref{lem:new-1-compat-ordering} can be used to characterize trivial yes-instances of $k$-\textsc{Compatible Ordering} as follows:

\begin{corollary}
    \label{cor:trivial-1-compat-ordering}
    Given an instance $(V, \mathcal{G})$ of $k$-\textsc{Compatible Ordering}, if there exists a pair of directed graphs $(A_i, B_i)$ in $\mathcal{G}$ whose union is a directed acyclic graph, $(V, \mathcal{G})$ is a yes-instance of $k$-\textsc{Compatible Ordering}.
\end{corollary}

The result from Corollary~\ref{cor:trivial-1-compat-ordering} can also be used to obtain an asymptotically optimal algorithm that can be used to check whether a pair $(\mathcal{S} = s_1, \ldots, s_n, \mathcal{L} = \ell_1, \ldots, \ell_n)$ is a solution to an instance of $k$-\textsc{Compatible Ordering} in time linear in the number of vertices and the number of constraint edges incident to the $s_i$'s in the directed graphs specified by their label, i.e., only the outgoing edges of $s_i$ in $A_{\ell_i}$ and the incoming edges of $s_i$ in $B_{\ell_i}$ are accounted for. This algorithm will be used as a subroutine in Theorems~\ref{thm:bdd-treewidth} and~\ref{thm:bdd-mdw}.

\begin{observation}
    \label{obs:residgraph}A pair $(\mathcal{S} = s_1, \ldots, s_n, \mathcal{L} = \ell_1, \ldots, \ell_n)$ is a solution to a $k$-\textsc{Compatible Ordering} instance $(V, \mathcal{G})$ if and only if the graph  $(V, \{s_is_j \mid s_is_j \in A_{\ell_i} \text{ or } s_is_j \in B_{\ell_j}\})$ is a directed acyclic graph.
\end{observation}

Given that Corollary~\ref{cor:trivial-1-compat-ordering} makes use of the acyclicity of the union of a specific pair of graphs $(A_i, B_i)$ in $\mathcal{G}$ to detect yes-instances of $k$-\textsc{Compatible Ordering}, it is natural to consider whether we can hope to solve instances that include cycles, whether in a given directed graph, in the union of two directed graphs $(A_i, B_i)$, or in the union of $\mathcal{G}$, in polynomial time. The answer to those questions is negative, and the relevant results are presented in Section~\ref{sec:hardness_results}. 

Moreover, these two lemmas do not leverage the geometry of the problems on robotic arms to its fullest extent, specifically when it comes to the positioning of the robotic arms in the plane. Indeed, if any large enough subset of the vertices of a $k$-\textsc{Compatible Ordering} instance encoding an \textsc{SM-RAMP} instance is considered, it is likely to not have many vertices interacting pairwise, as many robotic arms being closely clustered together is a characteristic of no-instances. In other words, in the union of $\mathcal{G}$, the size of the largest clique, that is, the clique number, is bounded. The treewidth of a graph is vaguely related to its clique number, in that the clique number is a lower bound on the treewidth. 

While it is true that, on general graphs, the gap between the two values can be arbitrarily large, we have reason to believe that the union of $\mathcal{G}$ in $k$-\textsc{Compatible Ordering} instances derived from \textsc{SM-RAMP} instances potentially has small (undirected) treewidth, owing to the underlying geometry of \textsc{SM-RAMP}, which motivates the next result.

The algorithm presented next takes in as input a tree decomposition of the undirected union of $\mathcal{G}$ with specific properties that simplify the description of the algorithm:

\begin{definition}[\cite{cygan:parameterized}]
    A \textbf{nice tree decomposition} is a rooted tree decomposition, where each node $N_i$ containing vertex set $V(N_i)$ has one of the following forms:
    \begin{enumerate}
        \item A leaf node: $N_i$ has no children, and the set of vertices that it contains has size 1.
        \item An introduce node: $N_i$ has one child $N_j$, and $V(N_i) = V(N_j) \cup \{v\}$ for some $v \not\in V(N_j)$.
        \item A forget node: $N_i$ has one child $N_j$, and $V(N_i) = V(N_j) \setminus \{v\}$ for some $v \not\in V(N_j)$.
        \item A join node: $N_i$ has two children $N_g$ and $N_h$, and $V(N_i) = V(N_g) = V(N_h)$.
    \end{enumerate}
\end{definition}

\thmTreewidth*

\begin{proof}
Given an instance $\mathcal{I} = (V, \mathcal{G})$ of $k$-\textsc{Compatible Ordering}, consider a nice tree decomposition of minimum width (equal to the treewidth $t$) of the undirected union of $\mathcal{G}$ with at most $4|V|$ nodes. Such a nice tree decomposition is guaranteed to exist, and can be found in $\mathcal{O}(|V|)$ time~\cite{bodlaender:treewidth}. We design a bottom-up dynamic programming algorithm on this tree decomposition to solve the problem. The tree decomposition is a rooted tree of nodes $N_i$, such that each node contains a set of vertices $V(N_i) \subseteq V$. We will use $V_T(N_i)$ to denote the union of the vertex sets of the descendants of $N_i$ in the tree, including $N_i$ itself, $\mathcal{I}_T(N_i)$ to denote the instance of $k$-\textsc{Compatible Ordering} induced on $V_T(N_i)$, and $\mathcal{I}(N_i)$ to denote the instance of $k$-\textsc{Compatible Ordering} induced on $V(N_i)$. We recall from the definition of a tree decomposition that for each vertex $v \in V$, the graph induced by the nodes containing $v$ forms a tree; this will be fundamental in the design of the algorithm.

For two orderings $\mathcal{S}, \mathcal{S}'$ such that the vertices in $\mathcal{S}$ are a subset of the vertices in $\mathcal{S}'$, a labeled ordering $(\mathcal{S}, \mathcal{L})$ is said to be \emph{compliant with} another labeled ordering $(\mathcal{S}', \mathcal{L}')$ if $\mathcal{S}$ is a subsequence of $\mathcal{S}'$, and the labeling of the vertices of $\mathcal{S}$ in $(\mathcal{S}, \mathcal{L})$ matches their labeling in $(\mathcal{S}', \mathcal{L}')$. The idea of the algorithm is to keep track of the labeled orderings $(\mathcal{S}, \mathcal{L})$ of $V(N_i)$ that are compliant with a solution to $\mathcal{I}_T(N_i)$ for each node $N_i$ in the tree decomposition, starting from the leaves.

We introduce a function $comp$ on three parameters: a node, an ordering of the vertices in the node, and a labeling of the vertices in the ordering. Formally, for node $N_i$ and a labeled ordering $(\mathcal{S}_i, \mathcal{L}_i)$ of $V(N_i)$, $comp(N_i, \mathcal{S}_i, \mathcal{L}_i) = 1$ if $(\mathcal{S}_i, \mathcal{L}_i)$ is compliant with a solution to $\mathcal{I}_T(N_i)$, and $comp(N_i, \mathcal{S}_i, \mathcal{L}_i) = 0$ otherwise. We introduce notation which simplifies denoting the addition and the removal of vertices from orderings. Given a labeled ordering $(\mathcal{S}_i, \mathcal{L}_i)$ , we use $\mathcal{S}_i - v$ to refer to the subsequence of $\mathcal{S}_i$ obtained after the removal of a vertex $v$ from $\mathcal{S}_i$, and we use $\mathcal{L}_i - 
 v$ to denote the subsequence of $\mathcal{L}_i$ obtained after the removal of the label of $v$ from $\mathcal{L}_i$. We now explain how the $comp$ values are computed, based on the type of the node in the nice tree decomposition of $G$. For each node $N_i$ in the explanation below, it is understood that the orderings $\mathcal{S}_i$ and the labelings $\mathcal{L}_i$ of the vertices in a node $N_i$ range over the $|V(N_i)|!$ possible orderings and the $k^{|V(N_i)|}$ possible labelings of the vertices in $N_i$ respectively:

\begin{itemize}
\item If the node $N_i$ is a leaf node, $comp(N_i, \mathcal{S}_i, \mathcal{L}_i) = 1$.
\item If the node $N_i$ is a forget node with child $N_j$ such that $V(N_j) \setminus V(N_i) = \{v\}$, then $comp(N_i, \mathcal{S}_i, \mathcal{L}_i)= 1$ if there exists a pair $(\mathcal{S}_j, \mathcal{L}_j)$ such that $comp(N_j, \mathcal{S}_j , \mathcal{L}_j) = 1$, $\mathcal{S}_i = \mathcal{S}_j - v$, and $\mathcal{L}_i = \mathcal{L}_j - v$, and otherwise $comp(N_i, \mathcal{S}_i, \mathcal{L}_i) = 0$.
\item If the node $N_i$ is a join node, with children $N_g$ and $N_h$, $comp(N_i, \mathcal{S}_i, \mathcal{L}_i) = 1$ if and only if $comp(N_g, \mathcal{S}_i, \mathcal{L}_i) = comp(N_h, \mathcal{S}_i, \mathcal{L}_i) = 1$, and otherwise $comp(N_i, \mathcal{S}_i, \mathcal{L}_i) = 0$.
\item If the node $N_i$ is an introduce node with child $N_j$ such that $V(N_i) \setminus V(N_j) = \{v\}$, then $comp(N_i, \mathcal{S}_i, \mathcal{L}_i) = 1$ if $comp(N_j, \mathcal{S}_j, \mathcal{L}_j) = 1$ where $\mathcal{S}_j = \mathcal{S}_i - v$, $\mathcal{L}_j = \mathcal{L}_i - v$, and $(\mathcal{S}_i, \mathcal{L}_i)$ is a solution to $\mathcal{I}(N_i)$, and otherwise $comp(N_i, \mathcal{S}_i, \mathcal{L}_i) = 0$. 
\end{itemize}

There are $\mathcal{O}(t! \cdot k^t)$ distinct Boolean values to be computed for each node $N_i$ in the tree decomposition, one for each labeled ordering of $V(N_i)$, which we refer to collectively as the $comp$ values of $N_i$. The $comp$ values of each internal node can be computed in time $\mathcal{O}(t^2 \cdot t! \cdot k^t)$, provided that the $comp$ values of its children are known. The analysis assumes that a value from $comp$ can be retrieved in constant time, which can be done by treating the parameters of $comp$ as indices by mapping each ordering and each labeling to an index based on their respective lexicographical ordering:

\begin{itemize}
    \item \textbf{Case 1:} The $comp$ values of leaf nodes can be computed in time $\mathcal{O}(k)$, since there are $k$ different labeled orderings of a set of size 1.
    
    \item \textbf{Case 2:} The $comp$ values of a forget node $N_i$ with child $N_j$ can be computed in time $\mathcal{O}(t \cdot t! \cdot k^t)$ by going through each of the $\mathcal{O}(t! \cdot k^t)$ labeled orderings $(\mathcal{S}_j, \mathcal{L}_j)$ of $V(N_j)$, removing $v$ to obtain $(\mathcal{S}_j - v, \mathcal{L}_j - v)$ in time $\mathcal{O}(t)$, and setting $comp(N_i, \mathcal{S}_j - v, \mathcal{L}_j - v) = comp(N_j, \mathcal{S}_j, \mathcal{L}_j)$ in constant time if $comp(N_i, \mathcal{S}_j - v, \mathcal{L}_j - v)$ has not been already assigned the value 1.
    
    \item \textbf{Case 3:} The $comp$ values of a join node $N_i$ with children $N_g$ and $N_h$ can be computed in time $\mathcal{O}(t! \cdot k^t)$ by going through each of the $\mathcal{O}(t! \cdot k^t)$ labeled orderings $(\mathcal{S}_i, \mathcal{L}_i)$ of $V(N_i)$, retrieving the two values $comp(N_g, \mathcal{S}_i, \mathcal{L}_i)$ and $comp(N_h$, $\mathcal{S}_i$, $\mathcal{L}_i)$ in constant time, and assigning a value to $comp(N_i, \mathcal{S}_i, \mathcal{L}_i)$ accordingly in constant time.
    
    \item     \textbf{Case 4:} The $comp$ values of an introduce node $N_i$ with child $N_j$ can be computed in time $\mathcal{O}(t^2 \cdot t! \cdot k^t)$. We go through each of the $\mathcal{O}(t! \cdot k^t)$ labeled orderings $(\mathcal{S}_i, \mathcal{L}_i)$ of $V(N_i)$ and we remove $v$ to obtain the labeled ordering $(\mathcal{S}_i - v, \mathcal{L}_i - v)$ in time $\mathcal{O}(t)$. If $comp(N_j, \mathcal{S}_i - v, \mathcal{L}_i - v) = 0$ or if $(\mathcal{S}_i, \mathcal{L}_i)$ is not a solution to $\mathcal{I}(N_i)$, we set $comp(N_i, \mathcal{S}_i, \mathcal{L}_i) = 0$; otherwise, we set $comp(N_i, \mathcal{S}_i, \mathcal{L}_i) = 1$. Both the construction of $\mathcal{I}(N_i)$ and verifying whether $(\mathcal{S}_i, \mathcal{L}_i)$ is a solution can be done in time linear in its size, i.e., in time $\mathcal{O}(t^2)$ (Observation~\ref{obs:residgraph}), since $\mathcal{I}(N_i)$ has $\mathcal{O}(t)$ vertices and $\mathcal{O}(t^2)$ arcs.
\end{itemize} 

Since there are a total of $\mathcal{O}(|V|)$ nodes, and each node's $comp$ values can be computed in time $\mathcal{O}(t^2 \cdot t! \cdot k^t)$, the $comp$ values of all the nodes in the tree decomposition can be computed in time $\mathcal{O}(|V| \cdot (t^2 \cdot t! \cdot k^t))$. We claim that the computed $comp$ values are correct. This would imply that $(V, \mathcal{G})$ is a yes-instance of $k$-\textsc{Compatible Ordering} if and only if $comp(N_r, \mathcal{S}_r, \mathcal{L}_r) = 1$ for some labeled ordering $(\mathcal{S}_r, \mathcal{L}_r)$ of $V(N_r)$, where $N_r$ is the root node of the tree decomposition. The proof we present next can be easily modified to also compute a solution $(\mathcal{S}, \mathcal{L})$ for the instance $(V, \mathcal{G})$.

We prove the correctness of the $comp$ values inductively with a case analysis based on the type of the node:

\begin{itemize}
    \item \textbf{Case 1:} The $comp$ values are correct for leaf node $N_i$, as $\mathcal{I}_T(N_i)$ has a single ordering, and is a yes-instance regardless of the label of the vertex.
    \item \textbf{Case 2:} The $comp$ values are correct for forget node $N_i$ with child $N_j$: since $V(N_i) \subset V(N_j)$, $\mathcal{I}_T(N_i)$ is equal to $\mathcal{I}_T(N_j)$. Moreover, if $(\mathcal{S}_j, \mathcal{L}_j)$ is compliant with a solution to $\mathcal{I}_T(N_j)$, then so is $(\mathcal{S}_j - v, \mathcal{L}_j - v)$ for any $v \in V(N_j)$. Finally, $(\mathcal{S}_i, \mathcal{L}_i)$ cannot be compliant with a solution to $\mathcal{I}_T(N_j)$ unless it can be extended by the insertion of $v$ (resp. a label of $v$) in some position in $\mathcal{S}_i$ (resp. $\mathcal{L}_i$), such that the resulting pair $(\mathcal{S}_j, \mathcal{L}_j)$ is compliant with a solution to $\mathcal{I}_T(N_j)$.
    \item \textbf{Case 3:} The proof of correctness is identical for the $comp$ values of join nodes and those of introduce nodes. We use join nodes to introduce the argument, before explaining how to apply it to introduce nodes. The argument hinges on the following observation:

    \begin{observation}
        \label{obs:core}
        Consider an instance $(V, \mathcal{G})$ of $k$-\textsc{Compatible Ordering}. Let $V_1, V_2 \subseteq V$ such that $V_1 \cup V_2 = V$, $V_3 = V_1 \cap V_2 \neq \emptyset$, and there do not exist arcs between $V_1 \setminus V_3$ and $V_2 \setminus V_3$ in $\mathcal{G}$. If $(\mathcal{S}_3, \mathcal{L}_3)$, a labeled ordering of $V_3$, is compliant with a solution $(\mathcal{S}_1, \mathcal{L}_1)$ to the instance induced on $V_1$ and a solution $(\mathcal{S}_2, \mathcal{L}_2)$ to the instance induced on $V_2$, then it is compliant with a solution to $(V, \mathcal{G})$. 
    \end{observation}

    \begin{proof}
        Since there are no arcs between vertices in $V_1 \setminus V_3$ and vertices in $V_2 \setminus V_3$, the relative ordering of their vertices in a solution to $(V, \mathcal{G})$ does not matter. 

        Given that $(\mathcal{S}_3, \mathcal{L}_3)$ is compliant with both $(\mathcal{S}_1, \mathcal{L}_1)$ and $(\mathcal{S}_2, \mathcal{L}_2)$, it will be compliant with a solution $(\mathcal{S}, \mathcal{L})$ to $(V, \mathcal{G})$, in which the vertices in $V_1 \setminus V_3$ (resp. $V_2 \setminus V_3$) are labeled with their labels in $\mathcal{L}_1$ (resp. $\mathcal{L}_2$), and are ordered with respect to the vertices in $\mathcal{S}$ as they are in $\mathcal{S}_1$ (resp. $\mathcal{S}_2$). \end{proof}

    For a join node $N_i$ with children $N_g$ and $N_h$, it is therefore sufficient to define the two vertex sets $V_1, V_2$. We make use of the following observation, which follows directly from the fact that the nodes containing a vertex $v$ induce a tree in a tree decomposition:
    \begin{observation}
    \label{obs:emptyinters}$(V_T(N_g) \setminus V(N_i)) \cap (V_T(N_h) \setminus V(N_i)) = \emptyset$, and there do not exist arcs between vertices in $V_T(N_g) \setminus V(N_i)$ and vertices in $V_T(N_h) \setminus V(N_i)$.
\end{observation}

    As such, we set $V_1 = V_T(N_g), V_2 = V_T(N_h)$. By Observation~\ref{obs:emptyinters}, all the preconditions of Observation~\ref{obs:core} hold; we can therefore apply Observation~\ref{obs:core} to prove the correctness of the $comp$ values of $N_i$ that are assigned the value 1.

    As for the $comp$ values of $N_i$ that are assigned the value 0, their proof of correctness is trivial: if $(\mathcal{S}_i, \mathcal{L}_i)$ is not compliant with a solution to $\mathcal{I}_T(N_g)$ and $\mathcal{I}_T(N_h)$, then it will not be compliant with a solution to $\mathcal{I}_T(N_i)$ either.

    \item \textbf{Case 4:} The proof of correctness for the $comp$ values of introduce nodes reuses the argument for those of join nodes. In particular, for introduce node $N_i$ with child $N_j$, we set $V_1 = V(N_i)$, $V_2 = V_T(N_j)$, and we replicate the rest of the argument to conclude that the $comp$ values of $N_i$ that are assigned the value 1 are correct. 

As for the $comp$ values of $N_i$ that are assigned the value 0, where $v = V(N_i) \setminus V(N_j)$, if a labeled ordering $(\mathcal{S}_i, \mathcal{L}_i)$ is not a solution to $\mathcal{I}(N_i)$, it will not be compliant with any solution to $\mathcal{I}_T(N_i)$. Otherwise, $(\mathcal{S}_j, \mathcal{L}_j) = (\mathcal{S}_i - v, \mathcal{L}_i - v)$ is not compliant with a solution to $\mathcal{I}_T(N_j)$, which implies that $(\mathcal{S}_j, \mathcal{L}_j)$ is not compliant with a solution to $\mathcal{I}_T(N_i)$ either. Since the vertices in $\mathcal{S}_j$ are a proper subset of the vertices in $\mathcal{S}_i$, we therefore conclude that $(\mathcal{S}_i, \mathcal{L}_i)$ is not compliant with a solution to $\mathcal{I}_T(N_i)$, as needed.

\end{itemize}

Given that the $comp$ values for all the nodes in the tree decomposition can be computed in time $\mathcal{O}(|V| \cdot (t^2 \cdot t! \cdot k^t))$, and that the $comp$ values of the root of the tree can be used to solve the original $k$-\textsc{Compatible Ordering} instance, we obtain an efficient bottom-up algorithm that runs on nice tree decompositions of instances where the undirected union of $\mathcal{G}$ has bounded treewidth.
\end{proof}

Theorem~\ref{thm:bdd-mdw} proposes an algorithm that runs on $k$-\textsc{Compatible Ordering} instances with bounded labeled modular width, which is a variant of the modular width measure. This variant is intended to work in contexts where edges are directed and labeled, as is the case for $k$-\textsc{Compatible Ordering}. The usefulness of the labeled modular width measure lies in the union of $\mathcal{G}$ being likely to be dense only around vertices representing robotic arms that are close to each other in the plane, with pairs of robotic arms selected from two sets of robotic arms that are far from each other in the plane not having any constraints associated with them. This makes it possible to efficiently encode the constraints via a recursive definition of the graph, which we call the labeled modular decomposition of the graph.

We first recall the classical definition of the modular decomposition of a graph, originally defined on undirected graphs. The modular decomposition of an undirected graph makes it possible to represent the graph by an algebraic expression consisting of a set of four operations: the isolated vertex operation, the disjoint union operation, the complete join operation, and the substitution operation~\cite{gajarsky:parameterized}. A graph $G$ has bounded modular width if it can be constructed using a combination of those four operations, such that the minimum number of operands used across the substitution operations is bounded. The modular width of a graph is a structural parameter of interest when it comes to the design of parameterized algorithms not only because it generalizes other graph properties that are indicative of the edge density of a graph, such as neighborhood density and twin cover, but also because the modular decomposition allows for a recursive definition of graphs. One might wonder if it is possible to disregard labels and consider only the modular width of the undirected union of $\mathcal{G}$, as we did for treewidth (which would also allow us to avoid redefining new notation). The answer is negative. Indeed, if we add an additional pair of graphs $(A_i,B_i)$ forming a complete bidirected clique to an instance of $k$-\textsc{Compatible Ordering}, then the modular width of the undirected union of $\mathcal{G}$ is equal to 1. However, the inclusion of this new pair does not make the instance any simpler to solve, because no more than one vertex can be assigned the new label. Labeled modular width is indeed needed to obtain a positive result.

Given the above discussion, we introduce variants of modular width and modular decomposition for $k$-\textsc{Compatible Ordering} instances, which accommodate arcs and retain information about the directed graphs containing each arc. To this end, we label each arc by a subset of the $A$ graphs and the $B$ graphs. We now have the context needed to introduce the needed definitions. 

\begin{definition}
    Consider the following four operations:
    
\begin{itemize}
    \item O1: create an isolated vertex;
    \item O2: form the \textbf{disjoint union} $G_3 = (V_3, E_3)$ of two graphs $G_1 = (V_1, E_1)$ and $G_2 = (V_2, E_2)$, where $V_3 = V_1 \cup V_2$ and $E_3 = E_1 \cup E_2$;
    \item O3: form the \textbf{labeled complete join} $G_3 = (V_3, E_3)$ of two labeled directed graphs $G_1 = (V_1, E_1)$ and $G_2 = (V_2, E_2)$, where $V_3 = V_1 \cup V_2$ and $E_3 = E_1 \cup E_2 \cup E_{1 \rightarrow 2} \cup E_{2 \rightarrow 1}$. The labeled arc set $E_{1 \rightarrow 2}$ (resp. $E_{2 \rightarrow 1}$)  is in reference to a set of labeled arcs from each vertex in $G_1$ (resp. $G_2$) to each vertex in $G_2$ (resp. $G_1$), where at least one of $E_{1 \rightarrow 2}$ and $E_{2 \rightarrow 1}$ is non-empty. Formally, given a fixed set of labels $L \subseteq \{\texttt{"}A_1\texttt{"}, \texttt{"}B_1\texttt{"}, \ldots, \texttt{"}A_k\texttt{"}, \texttt{"}B_k\texttt{"}\}$, $E_{1 \rightarrow 2} = \{(uv, L) : u \in V_1, v \in V_2\}$; $E_{2 \rightarrow 1}$ is defined analogously; and
    \item O4: form the \textbf{labeled substitution operation} for labeled directed graphs $G_1, \ldots, G_p$ with respect to a labeled directed graph $G_t$ with vertex set $v_1, \ldots, v_p$ (the \textbf{template graph}): the labeled substitution operation applied to graphs $G_1 = (V_1, E_1), \ldots, G_p = (V_p, E_p)$, using graph $G_t$ as a template, yields a graph $G$, such that $V(G) = \bigcup\limits_{i = 1}^p V_i$ and $E(G) = (\bigcup\limits_{i = 1}^p E_i) \cup E'$, where $E' = \{(u_iu_j, L) : u_i \in G_x, u_j \in G_y, (v_xv_y, L) \in E(G_t)\}$. We note that O3 is a special case of O4 where $p = 2$.
\end{itemize}

The notions equivalent to modular width and modular decomposition, where a labeled directed graph $G$ is constructed using the four modified operations above, are called the \textbf{labeled modular width} and the \textbf{labeled modular decomposition} of $G$ respectively.
\end{definition}

Naturally, labeled graphs also make it possible to look at an analogue of the union operation that incorporates labels.

\begin{definition}
    The \textbf{labeled union} of $\mathcal{G}$ is its union augmented with the addition to each arc of a set of labels that is a subset of $\{\texttt{"}A_1\texttt{"}, \texttt{"}B_1\texttt{"}, \ldots, \texttt{"}A_k\texttt{"}, \texttt{"}B_k\texttt{"}\}$, denoting the $A$ graphs and the $B$ graphs that contain the edge.
\end{definition}

For brevity, we use $mw(\mathcal{G})$ to refer to the labeled modular width of the labeled union of $\mathcal{G}$, with the name of the collection of directed graphs omitted when clear from context. While it is known that computing the modular width and a modular decomposition can be done in linear time for general graphs~\cite{mcconnell:moddecomp,tedder:moddecomp}, it is not possible to replicate the proof techniques to obtain an efficient algorithm to compute their labeled variants, so the next result assumes that a labeled modular decomposition of the labeled union of $\mathcal{G}$  is given as part of the input.

\thmModularwidth*

\begin{proof}
We design a dynamic programming algorithm, which assumes that a labeled directed modular decomposition is provided alongside the input. For each labeled directed graph that is an operand of one of the four operations, we consider the $k$-\textsc{Compatible Ordering} instance induced on its vertices, and we keep track of the different sets of labels, or \emph{label selections} (not to be confused with labelings) that can be used to solve it. Given the instance $(V, \mathcal{G})$ of $k$-\textsc{Compatible Ordering}, let $G^f$ represent the labeled union of the directed graphs in $\mathcal{G}$. For conciseness, since $G^f$ encodes the instance, whenever we mention solving $G^f$, we are referring to solving the instance $(V, \mathcal{G})$. Let $\mathbb{L}(G^f)$ denote the set of label selections that can be used to solve $G^f$. The argument relies on induction on the number of vertices in $G^f$; we give a case analysis based on the four operations that could have been used to obtain the graph $G^f$. 

The analysis for the first two operations is trivial. If $G^f$ is a graph on a single vertex, this single vertex is the only vertex in an ordering that solves the instance, therefore $\mathbb{L}(G^f) = \{\{1\}, \ldots, \{k\}\}$. If $G^f$ is the disjoint union of two labeled directed graphs $G^f_1$ and $G^f_2$, $\mathbb{L}(G^f)$ is equal to the pairwise union of the label selections in $\mathbb{L}(G^f_1)$ and those in $\mathbb{L}(G^f_2)$, because the vertex sets of $G^f_1$ and $G_2^f$ are independent. Next, we handle the case where $G^f$ is obtained via an application of O4, which also handles O3.

Suppose that $G^f$ is obtained via an application of O4, by substituting the labeled directed graphs $G^f_1, \ldots, G^f_p$ into the template graph $G_t$ that has vertex set $v_1, \ldots, v_p$. To compute which label selections can be used to solve $G^f$, we look at the template graph, we go over all possible label selections of the $p$ instances, and we check which combinations of label selections allow for an ordering of the vertices of $G_t$. Formally, for each element $(L_1, \ldots, L_p)$ in $\mathbb{L}(G^f_1) \times \ldots \times \mathbb{L}(G^f_p)$, we update the neighborhood of each vertex $v_j$ in $V(G_t)$ to preserve $v_j$'s outgoing and incoming arcs whose label set shares at least one of the labels in $L_j$, and delete the rest of the arcs. Then, we check whether $G_t$ is a directed acyclic graph. If it is, $L_1 \cup \ldots \cup L_p$ is a label selection that can be used to solve the instance induced on the vertices of $G^f$; if it is not, then it is not possible to solve the instance induced on the vertices of $G^f$ using only labels from $L_1 \cup \ldots \cup L_p$ (Observation~\ref{obs:residgraph}). 

The number $p$ of graphs that are substituted in the template graph is at most $mw$ by definition, each of those $p$ graphs has at most $2^k$ possible sets of labels, and we want to try all $2^{k + mw}$ sets of label selections to see which ones can be used to solve $G^f$. After the set of label selections is determined, the work that is left to be done is updating $G_t$, checking for acyclicity, and computing the union of the label selections, all of which can be executed in time polynomial in the size of the input. At most $|V|$ operations are applied, since each operation increases the number of vertices by at least one, and this gives us the desired running time of $f(k + mw) \cdot |V|^{\mathcal{O}(1)}$.
\end{proof}

\subsection{Hardness results}
\label{sec:hardness_results}

The instances of $k$-\textsc{Compatible Ordering} problem that we were efficiently able to solve so far (using Corollary~\ref{cor:trivial-1-compat-ordering}) exhibited acyclicity within their input. In this section, we restrict the presence of cycles to represent cyclical constraints that naturally arise in the context of the \textsc{RAMP} problems. We prove below that $k$-\textsc{Compatible Ordering} is hard under each of those restrictions.

The hardness results covered in this section stem from restricted versions of the \textsc{SAT} problem, namely the \textsc{Planar 4-Bounded 3-SAT} problem. Let $\varphi$ be a \textsc{3-SAT} instance, and let $G(\varphi)$ be a graph representation of $\varphi$, defined as follows: for each variable and clause in $\varphi$, we add a vertex to the vertex set of $G(\varphi)$. As for the edge set of $G(\varphi)$, a clause vertex is adjacent to a variable vertex if and only if the clause contains a literal of that variable. If $G(\varphi)$ is planar, and each variable occurs at most 4 times in $\varphi$, irrespective of the polarity of its literals, we say that $\varphi$ is a \textsc{Planar 4-Bounded 3-SAT} instance. 
To prove that $k$-\textsc{Compatible Ordering} is $\mathsf{NP}$-hard, we reduce from \textsc{Planar 4-Bounded 3-SAT}, which is known to be $\mathsf{NP}$-hard, even on instances where each clause contains exactly 3 literals~\cite{tippenhauer:3sat}.

\begin{figure*}[ht]
    \centering
    \includegraphics[width=\linewidth]{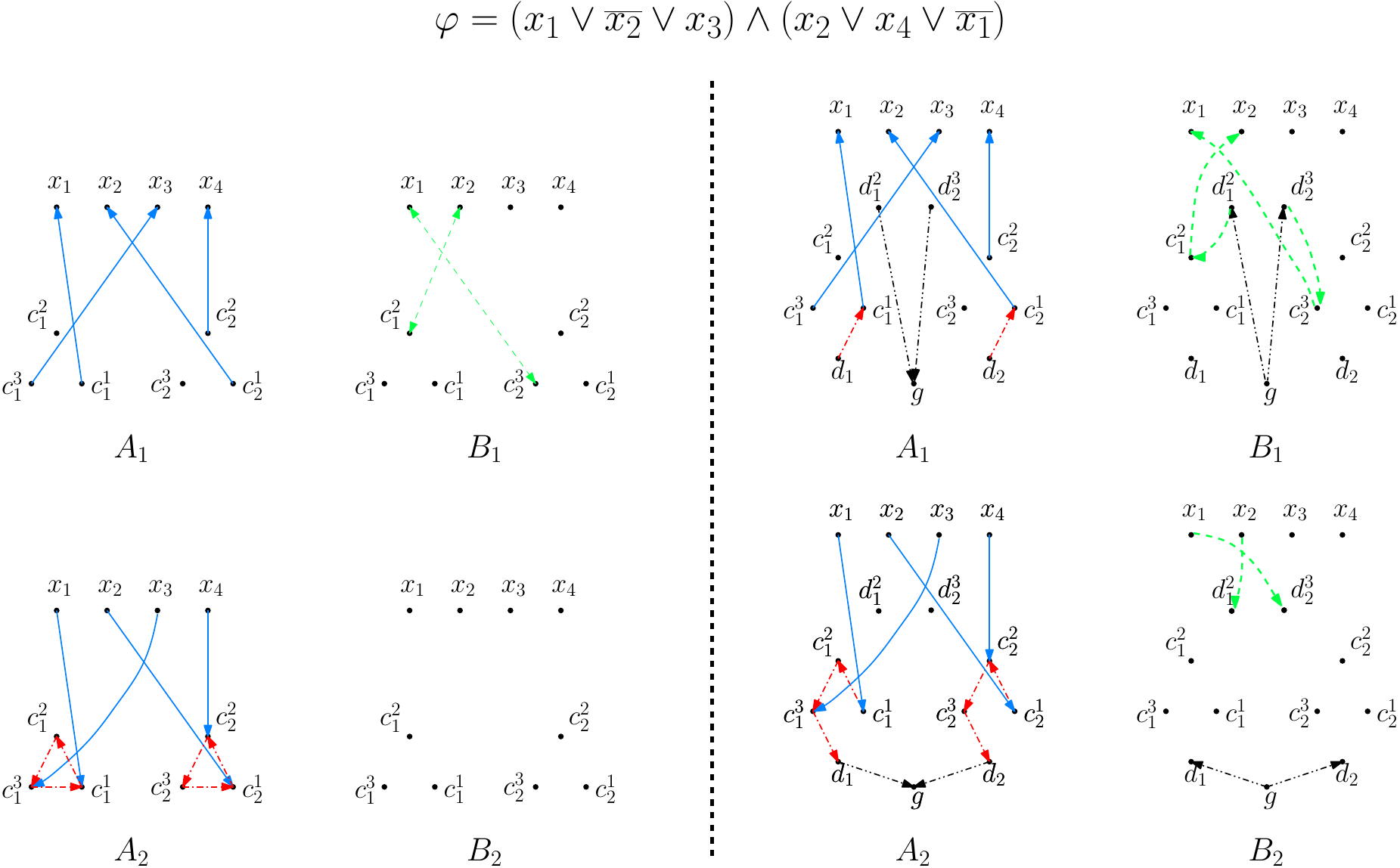}
    \caption{Illustration of the reductions from Theorems~\ref{thm:kcompat-npc} (left) and~\ref{thm:kcompat-npc-acyclic} (right), starting from the same \textsc{Planar 4-Bounded 3-SAT} instance $\varphi$ (top). The labeled ordering $(\mathcal{S}_1 = x_1, x_2, c_1^1, c_2^1, c_1^3, c_1^2, c_2^3, c_2^2, x_3, x_4, \mathcal{L}_1 = 1, 1, 1, 1, 2 ,2, 2, 2, 2, 2)$ is a solution to the $2$-\textsc{Compatible Ordering} instance on the left, whereas the labeled ordering $(\mathcal{S}_2 = d_2^3, x_1, d_1^2, x_2, d_1, c_1^1, d_2, c_2^1, c_1^3, c_1^2, c_2^3, c_2^2, x_3, x_4, \mathcal{L} = 2, 1, 2, 1, 1, 1, 1, 1, 2 ,2, 2, 2, 2, 2)$ is a solution to the $2$-\textsc{Compatible Ordering} instance on the right.}
    \label{fig:3-sat-kco-reduction}
\end{figure*}

\thmNPPlanar*


\begin{proof}
We show that $2$-\textsc{Compatible Ordering} is $\mathsf{NP}$-complete under the stated conditions; this suffices since $2$-\textsc{Compatible Ordering} is a special case of $k$-\textsc{Compatible Ordering}.

Membership in $\mathsf{NP}$ follows from the fact that a labeled ordering $(\mathcal{S}, \mathcal{L})$ serves as a certificate; verifying all constraints in $\mathcal{O}(n + m)$ time is straightforward (Observation~\ref{obs:residgraph}).

Let $\varphi = C_1 \wedge \cdots \wedge C_p$ be a \textsc{Planar 4-Bounded 3-SAT} instance on variables $X=\{x_1,\dots,x_q\}$. We construct an instance $(V, \{(A_1,B_1),(A_2,B_2)\})$ of $2$-\textsc{Compatible Ordering} as follows:
\smallskip
\\
\\
\noindent
 \emph{Vertices in $V$:}
\begin{itemize}
    \item For each clause $C_i$, add three \emph{clause vertices} $c_i^1, c_i^2, c_i^3$ to $V$, each representing one literal of $C_i$.
    \item For each variable $x_j \in X$, add a \emph{variable vertex} to $V$.
\end{itemize}

The vertex set $V$ of the constructed instance is therefore of size $3p + q$. Given a clause vertex $c_i^a$ and a variable $x_j$ such that the literal that $c^a_i$ represents is a literal of the variable $x_j$, we say that the \emph{associated variable} of $c_i^a$ is $x_j$.

In the constructed instance, assigning a variable vertex the label 1 (resp. the label 2) corresponds to the variable being set to \emph{true} (resp. \emph{false}), and assigning a clause vertex the label 1 corresponds to the literal it represents being set to \emph{true}; we therefore need to ensure that at least one clause vertex per clause is assigned the label 1, and that the labeling of the variable vertex associated with a clause vertex labeled 1 is consistent, i.e., the variable vertex is assigned label 1 (resp. label 2) if the clause vertex is a positive (resp. negative) clause vertex. The arc set of the instance ensures consistency holds. 
\noindent
\\
\\
\emph{Arcs in $A_1,B_1,A_2,B_2$:}

\begin{itemize}
    \item \textbf{Clause cycles:} For each clause $C_i$, we add the directed cycle $c_i^1c_i^2$, $c_i^2c_i^3$, $c_i^3c_i^1$ to $A_2$.
    \item \textbf{Positive clause-variable arcs:} For each pair $c^+, x$, where $c^+$ is a positive clause vertex and $x$ is the associated variable of $c^+$, we add the arc $c^+x$ to $A_1$, and the arc $xc^+$ to $A_2$. 
    \item \textbf{Negative clause-variable cycles:} For each pair $c^-, x$, where $c^-$ is a negative clause vertex and $x$ is the associated variable of $c^-$, we add the arcs $c^-x$ and $xc^-$ to $B_1$. 
\end{itemize}

We remark that there are no arcs in $B_2$, and that the variable vertices form independent sets in all four directed graphs. An example depicting the different types of arcs in the construction can be found in Figure~\ref{fig:3-sat-kco-reduction}.

We claim that $\varphi$ is satisfiable if and only if $(V,\mathcal{G} = \{(A_1, B_1), (A_2, B_2)\})$ is a yes-instance of $2$-\textsc{Compatible Ordering}.
\\
\\
\noindent
\textbf{($\Rightarrow$)}
If $\varphi$ has a satisfying assignment, let $\mathcal{X}^T$ (resp. $\mathcal{X}^F$) be an arbitrary ordering of the variable vertices representing variables that are assigned \emph{true} (resp. \emph{false}) in the assignment, and let $\mathcal{C}^T$ be an arbitrary ordering of $q$ clause elements, one per clause, whose literals evaluate to \emph{true}: each clause $C_i$ therefore has two clause elements $c^a_i$ and $c^b_i$ that are not in $\mathcal{C}^T$. An ordering $\mathcal{C}^R_i$ of those two clause elements as given as follows: $c_i^a$ comes before $c_i^b$ in $\mathcal{C}^R_i$ if $c_i^ac_i^b$ is an edge in $A_2$, and $c_i^a$ comes after $c_i^b$ otherwise. Let $\mathcal{C}^R = \mathcal{C}^R_1, \ldots, \mathcal{C}^R_p$. 

We set $\mathcal{S} = \mathcal{X}^T, \mathcal{C}^T, \mathcal{C}^R, \mathcal{X}^F$. We assign label 1 to the vertices in $\mathcal{X}^T$ and $\mathcal{C}^T$, we assign label 2 to the remaining vertices, and we denote by $\mathcal{L}$ the labeling of $\mathcal{S}$. We claim that $(\mathcal{S}, \mathcal{L})$ is a valid labeled ordering. We include a brief proof of correctness below:
\begin{itemize}
    \item \textbf{Case $\mathcal{X}^T$:} Each variable vertex in $\mathcal{X}^T$ has no out-neighbors in $A_1$, and its in-neighbors in $B_1$, which are clause vertices, come after it in $\mathcal{S}$. 
    \item \textbf{Case $\mathcal{C}^T$:} Each positive clause vertex in $\mathcal{C}^T$ has a unique out-neighbor in $A_1$ in $\mathcal{X}^T$ which comes before it in $\mathcal{S}$, and has no in-neighbors in $B_1$. Each negative clause vertex in $\mathcal{C}^T$ has no out-neighbor in $A_1$, and has a unique in-neighbor in $B_1$ in $\mathcal{X}^F$, which comes after it in $\mathcal{S}$.
    \item \textbf{Case $\mathcal{C}^R$:} The vertices in $\mathcal{C}^R$ require a slightly more careful analysis. Consider an arbitrary clause vertex $c_i^a$ in $\mathcal{C}^R_{i}$: $c_i^a$ has a unique out-neighbor $c_i^b$ in $A_2$, which is either in $\mathcal{C}^T$ and has been previously included in $\mathcal{S}$, or is in $\mathcal{C}^R_{i}$ and has been included in $\mathcal{S}$ directly before $c^a_i$ by construction, and has no in-neighbors in $B_2$, since there are no arcs in $B_2$.
    \item \textbf{Case $\mathcal{X}^F$:} Each variable vertex in $\mathcal{X}^F$ has no in-neighbors in $B_2$, since there are no arcs in $B_2$, and its out-neighbors in $A_2$, which are in either $\mathcal{C}^T$ or $\mathcal{C}^R$, come before it in $\mathcal{S}$.
\end{itemize}

The proposed labeled ordering $(\mathcal{S}, \mathcal{L})$ is therefore valid.

\noindent
\\
\textbf{($\Leftarrow$)} Consider the three clause vertices $c_i^1, c_i^2, c_i^3$ of an arbitrary clause $C_i$. In a solution to $k$-\textsc{Compatible Ordering}, the following observation holds:

\begin{observation}
    \label{obs2} Each positive (resp. negative) clause vertex in $C_i$ that is included in the ordering and assigned label 1 comes before (resp. after) its associated variable vertex, which is assigned the label 1 (resp. 2), in the ordering.
\end{observation}

\begin{proof}
Let $c_i^a$ be a clause vertex labeled 1, and let $x_j$ be the associated variable. If $c_i^a$ is a positive clause vertex, given the arc $c_i^ax_j$ in $A_1$, $x_j$ must come before $c_i^a$ (Observation~\ref{obs:ordering}). Moreover, since we can now assume that $x_j$ comes before $c_i^a$, $x_j$ cannot be assigned label 2 due to the arc $x_jc^a_i$ in $A_2$. Otherwise, if $c_i^a$ is a negative clause vertex, given the arc $x_jc^a_i$ in $B_1$, $x_j$ must come after $c_i^a$ (Observation~\ref{obs:ordering}). Moreover, since we can now assume that $x_j$ comes after $c_i^a$, $x_j$ cannot be assigned label 1 due to the arc $c^a_ix_j$ in $B_1$.
\end{proof}

Let $C^{\ell = 1}$ be the set of clause vertices that are assigned label 1. It follows from Observation~\ref{obs2} that $C^{\ell = 1}$ cannot include negative and positive clause vertices with the same associated variable, as this would imply that the associated variable is included multiple times in the ordering, and is assigned two different labels. We partition $C^{\ell = 1}$ into two sets of clause vertices, $C^{+, \ell = 1}$ and $C^{-, \ell = 1}$, based on polarity, i.e., $C^{+, \ell = 1}$ (resp. $C^{-, \ell = 1}$) is the set of positive (resp. negative) clause vertices in $C^{l = 1}$. Let $X^{+, \ell = 1}$ (resp. $X^{-, \ell = 1}$) be the set of variable vertices associated with the clause vertices in $C^{+, \ell = 1}$ (resp. $C^{-, \ell = 1}$). Note that $|X^{+, \ell = 1}| \leq |C^{+, \ell = 1}|$ and $|X^{-, \ell = 1}| \leq |C^{-, \ell = 1}|$, since multiple clause vertices may be associated with the same variable vertex, and that $X^{+, \ell = 1} \cap X^{-, \ell = 1} = \emptyset$, which follows from Observation~\ref{obs2}.

We observe that $X \setminus \left(X^{+, \ell = 1} \bigcup X^{-, \ell = 1}\right)$ need not be empty. We claim that setting each variable in $X^{+, \ell = 1}$ to \emph{true} and setting each variable in $X^{-, \ell = 1}$ to \emph{false} satisfies $\varphi$; the variables in $X \setminus \left(X^{+, \ell = 1} \bigcup X^{-, \ell = 1}\right)$ are \emph{don't-care} variables that can be set to either \emph{true} or \emph{false}. Suppose otherwise: there exists a clause $C_i$ in $\varphi$, that is not satisfied by the proposed truth value assignment. By Observation~\ref{obs:directedcycle}, a clause vertex $c_i^a$ of $C_1$, whose associated variable is $x_j$, is included in the ordering with label 1. Furthermore, by Observation~\ref{obs2}, if $c^a_i$ is a positive (resp. negative) clause vertex, $x_j \in X^{+, \ell = 1}$ (resp. $x_j \in X^{-, \ell = 1}$), and is assigned the truth assignment \emph{true} (resp. \emph{false}) by our construction, which satisfies $c_i^1$, and therefore $C_i$, a contradiction. This concludes the proof of correctness of the reduction. 

Since $\varphi$ is planar 4-bounded, $G(\varphi)$ has a planar embedding. Replacing each clause vertex with three closely spaced points (for $c_i^1,c_i^2,c_i^3$) introduces no crossings, so the union remains planar. For degeneracy, each clause vertex has degree at most 3 in the underlying undirected graph without parallel edges (two edges in the cycle and one edge to the corresponding variable). Variable vertices form independent sets, implying any induced subgraph has a vertex of degree at most 3. Hence the union of $A_1,A_2,B_1,B_2$ has degeneracy at most 3. We note that the maximum degree of a vertex in the union is $8$, which occurs when a variable appears $4$ times in $4$ distinct clauses, irrespective of polarity.

Thus, $2$-\textsc{Compatible Ordering} (and hence $k$-\textsc{Compatible Ordering}$)$ is $\mathsf{NP}$-complete under planarity and degeneracy~$3$.
\end{proof}

The construction covered in Theorem~\ref{thm:kcompat-npc} goes against the spirit of the problem, because it relies on the existence of directed digons in $B_1$ and directed triangles in $A_2$: directed cycles trivially prevent the existence of a solution where the cycle's vertices are all assigned the same label (Observation~\ref{obs:directedcycle}). One might therefore wonder whether forbidding directed cycles within each of the directed graphs makes the problem easier to solve. Our next result shows that the answer is negative. 

\thmNPAcyclic*

\begin{proof}
Membership of the problem in $\mathsf{NP}$ is immediate. We show hardness via a reduction from \textsc{Planar 4-Bounded 3-SAT}, modifying the construction used in the proof of Theorem~\ref{thm:kcompat-npc} to eliminate directed cycles in $A_2$ and $B_1$.

Let $\varphi = C_1 \wedge \cdots \wedge C_p$ be a \textsc{Planar 4-Bounded 3-SAT} instance on variables $X = \{x_1, \dots, x_q\}$. Starting with the vertex set $V$ in Theorem~\ref{thm:kcompat-npc} (which includes variable vertices and clause vertices), we add:
\begin{itemize}
  \item A \emph{dummy clause vertex} $d_i$ for each clause $C_i$.
  \item A \emph{dummy negative clause vertex} $d_i^a$ for each negative literal $c_i^a$.
  \item A single \emph{global vertex} $g$.
\end{itemize}
After the inclusion of the vertices above, $|V| = 4p + q + o + 1$, where $o$ is the total number of negative literals in $\varphi$.
\\
\\
We then define arcs in $A_1, B_1, A_2, B_2$ so that none of the four graphs contains a directed cycle, while preserving the logical structure of the previous reduction:

\begin{itemize}
\item \textbf{Clause arcs:} For each clause $C_i$ whose clause vertices are $c^1_i, c^2_i, c^3_i$, we add the arcs $c_i^1c_i^2, c_i^2c_i^3, c_i^3d_i$ to $A_2$, and the arc $d_ic_i^1$ to $A_1$. The clause cycles previously contained in $A_2$ in the construction of Theorem~\ref{thm:kcompat-npc} are now in the union of $A_1$ and $A_2$.
\item \textbf{Positive clause-variable arcs:} 
  As in Theorem~\ref{thm:kcompat-npc}, for a positive literal $c^+$ whose associated variable is $x$, add the arc $c^+x$ to $A_1$ and the arc $xc^+$ to $A_2$.
\item \textbf{Negative clause-variable arcs:}
  For each negative literal $c_i^a$ whose associated variable is $x$, add the arcs
    $d_i^ac_i^a$, $c_i^ax$ to $B_1$ and the arc $xd_i^a$ to $B_2$. The negative clause-variable cycles previously contained in $B_1$ in the construction of Theorem~\ref{thm:kcompat-npc} are now in the union of $B_1$ and $B_2$.
\item \textbf{Arcs involving the global vertex $g$:}  For each dummy clause vertex $d_i$, add the arc $d_ig$ to $A_2$ and the arc $gd_i$ to $B_2$.  For each dummy negative clause vertex $d_i^a$, add the arc $d_i^ag$ to $A_1$ and the arc $gd_i^a$ to $B_1$. These arcs guarantee dummy vertices obtain particular labels, as explained in Observation~\ref{obs:forbiddenlabel}, and prevent cycles involving $g$.
\end{itemize}

An example depicting the different types of arcs in the construction can be found in Figure~\ref{fig:3-sat-kco-reduction}.

We claim that  $\varphi$ is satisfiable if and only if $(V,\mathcal{G} = \{(A_1,B_1), (A_2,B_2)\})$ is a yes-instance of $2$-\textsc{Compatible Ordering}.

\noindent
\\
\textbf{($\Rightarrow$)}
If $\varphi$ is satisfiable, we begin with the labeled ordering from Theorem~\ref{thm:kcompat-npc} that arranges variable vertices and clause vertices consistently (ensuring at least one literal per clause is labeled 1), with the labeling $\mathcal{L}$ and the orderings $\mathcal{S}, \mathcal{X}^T, \mathcal{C}^T, \mathcal{C}^R$ and $\mathcal{X}^F$  retaining their definitions. We then insert:
\begin{itemize}
\item each $d_i$ with label 1 immediately after $c_i^1$ (resp. before $c_i^3$) if $c_i^1$ comes before (resp. after) $c_i^3$ in $\mathcal{S}$ to satisfy the constraints implied by the new arcs in $A_1$ and $A_2$,
\item each $d_i^a$ with label 2 immediately after $c_i^a$ (resp. before its associated variable $x$) if $c_i^a$ comes before (resp. after) $x$ in $\mathcal{S}$ to satisfy the constraints implied by the new arcs in $B_1$ and $B_2$, and
\item the global vertex $g$, labeled 1, at the end of the ordering.
\end{itemize}
It is easy to verify that the sink and source constraints of all the dummy vertices, as well as $g$, are satisfied.  The proof of correctness to come will therefore cover the rest of the vertex types.
\\
\\
Let $\mathcal{S}' = \mathcal{X}^T, \mathcal{C}^T_d, \mathcal{C}^R_d, \mathcal{X}^F_d, g$ and $\mathcal{L}'$ denote the updated ordering and labeling, with $\mathcal{C}^T_d, \mathcal{C}^R_d$ and $\mathcal{X}^F_d$ and denoting the augmentation of $\mathcal{C}^T, \mathcal{C}^R$ and $\mathcal{X}^F$ respectively with the dummy vertices, as explained earlier. 

We claim that $(\mathcal{S}', \mathcal{L}')$ is a valid labeled ordering. We replicate the structure of the proof of Theorem~\ref{thm:kcompat-npc}:
\begin{itemize}
    \item Each variable vertex in $\mathcal{X}^T$ has no out-neighbors in $A_1$, and its in-neighbors in $B_1$, which are clause vertices, come after it in $\mathcal{S}'$.
\item The vertices in $\mathcal{C}^T_d$ are dummy clause vertices, dummy negative clause vertices, positive clause vertices or negative clause vertices; we handle the last two types of vertices. Each positive clause vertex in $\mathcal{C}^T_d$ has a unique out-neighbor in $A_1$ in $\mathcal{X}^T$ which comes before it in $\mathcal{S'}$, and has no in-neighbors in $B_1$. Each negative clause vertex in $\mathcal{C}^T_d$ has no out-neighbors in $A_1$, and has a unique in-neighbor in $B_1$ in $\mathcal{X}^F$, a dummy negative clause vertex, which comes after it in $\mathcal{S}'$.
\item The vertices in $\mathcal{C}^R_d$ are dummy clause vertices, dummy negative clause vertices, positive clause vertices, or negative clause vertices; we handle the last two types of vertices. Consider an arbitrary clause vertex $c_i^a$ in $\mathcal{C}^R_d$. If $a = 1$ or $a = 2$, $c_i^a$ has a unique out-neighbor in $A_2$, which is either in $\mathcal{C}^T_d$ and comes before it in $\mathcal{S}'$, or is in $\mathcal{C}^R_d$ and has been included in $\mathcal{S}'$ directly before $c_i^a$. If $a = 3$, $c_i^a$ has a unique out-neighbor $d_i$ in $A_2$, which either comes directly before $c_i^3$ in $\mathcal{S}'$ if $c^3_i$ comes before $c_i^1$ in $\mathcal{S}'$, or comes directly after $c_i^1$, and therefore, before $c_i^3$, if $c_i^3$ comes after $c_i^1$ in $\mathcal{S}'$. In both cases, $c_i^a$ has no in-neighbors in $B_2$.
\item The vertices in $\mathcal{X}^F_d$ are either dummy negative clause vertices or variable vertices; we handle the second type of vertex. Each variable vertex in $\mathcal{X}^F_d$ has no in-neighbors in $B_2$, and their out-neighbors in $A_2$, which are clause vertices, come before them in $\mathcal{S}'$.
\item The global vertex $g$ has no out-neighbors in $A_1$, and no in-neighbors in $B_1$.
\end{itemize}

The proposed labeled ordering $(\mathcal{S}', \mathcal{L}')$ is therefore valid.

\noindent
\\
\textbf{($\Leftarrow$)}
The proof for the backward direction of Theorem~\ref{thm:kcompat-npc} relied on the fact that one of $c_i^1, c_i^2$ and $c_i^3$ is assigned label 1 for each clause $C_i$, as well as Observation~\ref{obs2}. We show that the two observations still hold under the new construction, which is sufficient to prove the existence of a satisfying assignment for $\varphi$.

To prove that, for each clause $C_i$, one of $c_i^1, c_i^2, c_i^3$ is assigned label 1 under the new construction, suppose, for the sake of contradiction, that there exists a clause $C_i$ whose clause vertices are all assigned label 2 in the ordering, which implies that its non-dummy clause vertices are in the following relative order: $c_i^3, c_i^2, c_i^1$. Since $d_i$ is assigned label 1 (Observation~\ref{obs:forbiddenlabel}), it will have to come after $c_i^1$, and therefore, after $c_i^3$, which violates a constraint on $c_i^3$ because of the arc $c_i^3d_i$ in $A_1$ (Observation~\ref{obs:ordering}).
As for Observation~\ref{obs2}, the previous proof is still applicable in the new construction, because the proof relies on the arcs constructed between variable vertices and non-dummy clause vertices, which were not modified.

The rest of the argument used to derive a satisfying assignment for $\varphi$ follows directly from the two statements above, and is identical to the proof of Theorem~\ref{thm:kcompat-npc}.

Thus, removing the cycles from $A_2$ and $B_1$ via dummy vertices still encodes satisfiability. Hence $2$-\textsc{Compatible Ordering} (and therefore $k$-\textsc{Compatible Ordering}$)$ remains $\mathsf{NP}$-complete when all graphs are acyclic.
\end{proof}

In the construction of Theorem~\ref{thm:kcompat-npc-acyclic}, cycles no longer exist in individual directed graphs, but rather, in the union of the directed graphs. By Corollary~\ref{cor:trivial-1-compat-ordering}, we already know that the problem admits a polynomial-time solution when the union of $\mathcal{G}$ is acyclic (and actually even when there exists a single $i$ for which $A_i \cup B_i$ is acyclic). So the natural next step would be looking for intermediate instances. Next, we prove that the problem remains $\mathsf{NP}$-complete even on instances whose union is acyclic after a simple transformation of the instance. Given a directed graph $G = (V, E)$, we define the \emph{reversed graph} of $G$ as the graph $(V, \{ v_jv_i$ $|$ $v_iv_j \in E\})$: we will reverse the $B$ graphs.

\begin{theorem}
    \label{appendix-thm:kcompat-dag}$k$-\textsc{Compatible Ordering} is $\mathsf{NP}$-complete on instances where the union of the $A$ graphs and the reversed $B$ graphs is acyclic.
\end{theorem}

\begin{proof} We have already proven in Theorem~\ref{thm:kcompat-npc} that the problem is in $\mathsf{NP}$. We prove that the problem is $\mathsf{NP}$-hard on the set of instances of interest via a reduction from \textsc{3-SAT} where each clause contains exactly three literals, which is still $\mathsf{NP}$-hard. 

Let $\varphi = C_1 \land \ldots \land C_p$ be a \textsc{3-SAT} instance on variables $X = \{x_1, \ldots, x_q\}$. The vertex set $V$ contains three vertex types:

\begin{itemize}
    \item Each variable in $X$ as a variable vertex,
    \item each clause in $\varphi$ as a clause vertex, and
    \item two vertices $t_1$ and $t_2$.
\end{itemize}

The vertex set $V$ of the constructed instance is therefore of size $p + q + 2$. The purpose of the two vertices $t_1$ and $t_2$ is to separate the variable vertices and the clause vertices in a solution. We now present the arcs that are included in the construction; this set of arcs will yield a set of directed graphs whose labeled union after reversal of the $B$ graphs is acyclic, and where the following sequence is a topological ordering of the vertices: $\mathcal{C}, \mathcal{X}, t_2, t_1$, where $\mathcal{C}$ and $\mathcal{X}$ are arbitrary orderings of the clause vertices and $X$ respectively.

\begin{itemize}
    \item \textbf{Arc from $t_2$ to $t_1$:} The arc $t_2t_1$ is added to $A_1, A_2, A_3$.
    \item  \textbf{Arcs between clause vertices and $t_1$ and $t_2$:} For each clause vertex $C_j$, we add the arc $C_jt_1$ (resp. $t_2C_j$) to $A_1, A_2, A_3$ (resp. $B_1, B_2, B_3$). 
    \item  \textbf{Clause-variable arcs:} For each clause $C_i$ whose clause vertices are $c_i^1, c_i^2, c_i^3$, and for each literal $c_i^a$ associated with the variable $x_j$, we add the arc $C_ix_j$ to $A_a$ (resp. the arc $x_jC_i$ to $B_a$) if and only if $c_i^a$ is a positive (resp. negative) literal in $C_i$. 
    \item \textbf{Arcs between variable vertices and $t_1$ and $t_2$}: For each variable vertex $x_j$ we add the arc $t_1x_j$ to $B_1$, the arcs $x_jt_1$ and $t_1x_j$ to $A_2$ and $B_2$ respectively, and the arc $x_jt_2$ to $A_3$.

\end{itemize}

\begin{figure*}
    \centering
    \includegraphics[width=\linewidth]{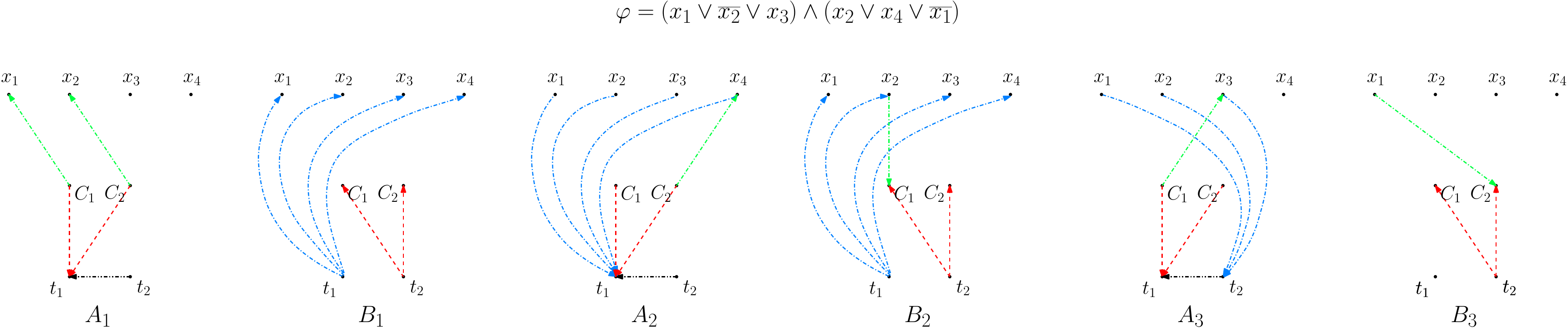}
    \caption{Illustration of the reduction from Theorem~\ref{thm:kcompat-npc-acyclic}.}
    \label{fig:reversed}
\end{figure*}

We remark that the variable vertices and the clause vertices form two independent sets in the union of the four digraphs. An example depicting the different types of arcs in the construction can be found in Figure~\ref{fig:reversed} (the ordering $\mathcal{S} = x_1, x_4, t_1, C_1, C_2, t_2, x_2, x_3$ solves the instance with labeling $\mathcal{L} = 1, 1, 1, 1, 2, 1, 3, 3$.).

We claim that $\varphi$ is a yes-instance of \textsc{3-SAT} if and only if $(V, \mathcal{G} = \{(A_1, B_1), (A_2, B_2)\})$ is a yes-instance of $3$-\textsc{Compatible Ordering}.
\\
\\
($\Rightarrow$) If $\varphi$ has a satisfying assignment, let $\mathcal{X}^T$ (resp. $\mathcal{X}^F$) be an arbitrary ordering of the variables that are assigned \emph{true} (resp. \emph{false}) in the assignment. We set $\mathcal{S} = \mathcal{X}^T, t_1, \mathcal{C}, t_2, \mathcal{X}^F$, where $\mathcal{C}$ is an arbitrary ordering of the clause vertices. We assign label 1 to $t_1$ and $t_2$, label 1 to each vertex in $\mathcal{X}^T$, and label 3 to each vertex in $\mathcal{X}^F$. The label that a clause variable $C_i$ is assigned will depend on which literal within $C_i$ is satisfied according to the chosen satisfying assignment. If $c^1_i$ is satisfied, $C_i$ is assigned label 1; if $c_i^2$ is satisfied, $C_i$ is assigned label 2, otherwise, $C_i$ is assigned label 3.

We claim that $(\mathcal{S}, \mathcal{L})$ is a valid labeled ordering. We include a brief proof of correctness:

\begin{itemize}
    \item Each variable vertex in $\mathcal{X}^T$ has no out-neighbors in $A_1$, and its in-neighbor in $B_1$, which is $t_1$, comes after it in $\mathcal{S}$.
    \item The vertex $t_1$ has no out-neighbors in $A_1$, and no in-neighbors in $B_1$.
    \item Each clause vertex in $\mathcal{C}$ has $t_1$ as an out-neighbor in the $A$ graphs, $t_2$ as an in-neighbor in the $B$ graphs, and for each pair $(A_l, B_l)$, either one variable vertex as an out-neighbor in $A_l$, or one variable vertex as an in-neighbor in $B_l$. Each clause vertex $C_i$ comes after $t_1$ and before $t_2$ in the ordering. Suppose that $C_i$ has been assigned label $l$: this implies that $c^l_i$, that has the associated variable $x_j$, evaluates to \emph{true}. If $c^l_i$ is a positive literal (resp. a negative literal), $C_ix_j$ (resp. $x_jC_i$) is an arc in $A_l$ (resp. $B_l$), and $x_j$ comes before (resp. after) $C_i$ in $\mathcal{S}$, because $x_j \in \mathcal{X}^T$ (resp. $x_j \in \mathcal{X}^F$), as needed. 
    \item The vertex $t_2$ has a unique out-neighbor in $A_1$, $t_1$, which comes before it in $\mathcal{S}$, and has no in-neighbors in $B_1$.
    \item Each variable vertex in $\mathcal{X}^F$ has a unique out-neighbor $t_2$ in $A_3$, which comes before it in $\mathcal{S}$, and has no in-neighbors in $B_3$.
\end{itemize}

The proposed labeled ordering $(\mathcal{S}, \mathcal{L})$ is therefore valid.
\\
\\
($\Leftarrow$) Let $(\mathcal{S}, \mathcal{L})$ be the valid labeled ordering that solves the constructed $3$-\textsc{Compatible Ordering} instance. In addition to the vertex $t_1$ coming before the vertex $t_2$ in $\mathcal{S}$, since the arc $t_2t_1$ exists in all $A$ graphs (Observation~\ref{obs:before}), $\mathcal{S}$ has the following characteristics:

\begin{observation}
    \label{obs:between}
    All the clause vertices come between $t_1$ and $t_2$ in $\mathcal{S}$.
\end{observation}
\begin{observation}
    \label{obs:outside}
    No variable vertex comes between $t_1$ and $t_2$ in $\mathcal{S}$.
\end{observation}

Observation~\ref{obs:between} is due to $t_1$ coming before $t_2$, and the arcs $C_it_1$ in the $A$ graphs, $t_2C_i$ in the $B$ graphs for each clause vertex $C_i$.
As for Observation~\ref{obs:outside}, for each variable vertex $x_j$, the arc $t_1x_j$ in $B_1$ means that $x_j$ must come before $t_1$ if $x_j$ is labeled 1 (Observation~\ref{obs:ordering}), the arcs $x_jt_1$ and $t_1x_j$ in $A_2$ and $B_2$ respectively mean that $x_j$ is never labeled 2 (Observation~\ref{obs:forbiddenlabel}), and the arc $x_jt_2$ in $A_3$ means that $x_j$ must come after $t_2$ if $x_j$ is labeled 3 (Observation~\ref{obs:ordering}).

In order to derive a satisfying assignment, we focus on the clause vertices, as well as the labels that they are assigned. Let $\ell_1, \ldots, \ell_p$ be the labels assigned to $C_1, \ldots, C_p$, respectively. Consider the set of literals $C^{sat} = \{c^{\ell_1}_1, \ldots, c^{\ell_p}_p\}$. We will partition those literals into two sets $C^T$ and $C^F$, based on whether the arc between $c_i^{\ell_i}$ and its associated variable, in one direction or the other, is in an $A$ graph or a $B$ graph, respectively. We then define $X^T$ (resp. $X^F$) as the set of variables in $\varphi$ with at least one literal in $C^T$ (resp. $C^F$). We observe that a variable vertex $x$ cannot belong to both $X^T$ and $X^F$, i.e., $X^T \cap X^F = \emptyset$, because an arc from a clause vertex in $C^{sat}$ to $x$ in an $A$ graph implies that $x$ comes before $t_1$ in $\mathcal{S}$, and an arc from the same clause vertex in $C^{sat}$ to $x$ in a $B$ graph implies that $x$ comes after $t_2$ in $\mathcal{S}$ (Observations~\ref{obs:ordering},~\ref{obs:outside}), a contradiction, since we already established that $t_1$ comes before $t_2$ in every valid labeled ordering. 

We claim that we can satisfy the 3-SAT formula $\varphi$ by setting all the variables in $X^T$ to \emph{true}, all the variables in $X^F$ to \emph{false}; the variables in $X \setminus \left(X^T \cup X^F\right)$ are \emph{don't-care} variables, and can be assigned either \emph{true} or \emph{false} (without loss of generality, we will set those to \emph{false}). Towards a contradiction, suppose that this proposed truth value assignment is incorrect, i.e., there exists a clause $C_i$ that is not satisfied. Its clause vertex has been assigned label $\ell_i$, and the variable associated with the literal $c_i^{\ell_i}$ is either in $X^T$ or in $X^F$, depending on whether it is a positive literal or a negative literal. If it is a positive literal, i.e., in $X^T$, we are setting it to true, whereas if it is a negative literal, i.e., in $X^F$, we are setting it to false, which satisfies $C_i$ in both cases.
\end{proof}

\section{Finding bounded $k$-compatible orderings}
\label{sec:bounded_ordering}
\vspace{1em}
In some cases, we may not need to order all of the vertices in a $k$-\textsc{Compatible Ordering} instance. Going back to the application of $k$-\textsc{Compatible Ordering} on moving objects from containers into trucks proposed in Section~\ref{sec:introduction}, the recipient of the delivery may be satisfied with receiving at least $b$ objects, as long as the objects are able to fit in the available trucks in some order. The problem definition can easily be modified to accommodate this change.
We introduce a bounded variant of the $k$-\textsc{Compatible Ordering} problem. In $k$-\textsc{Compatible Bounded Ordering}, a natural number $b$ is included as part of the input, for each vertex the source and sink constraints are defined identically to those in $k$-\textsc{Compatible Ordering}, and it suffices to find a valid labeled ordering of size at least $b$, where the source and sink constraints that should be satisfied only concern the vertices in the ordering. The problem is defined formally below:
\begin{tcolorbox}[colback=white, colframe=black,  
                  arc=4mm, boxrule=0.3mm, 
                  fonttitle=\bfseries]
$k$-\textsc{Compatible Bounded Ordering}
\\
\textbf{Input:} A vertex set $V$, a collection $\mathcal{G}$ of $k$ pairs of directed graphs $(A_1, B_1), \ldots,(A_k, B_k)$, such that $V(A_i) = V(B_i) = V$ for $1 \leq i \leq k$, a natural number $b$.
\\
\textbf{Output:} Whether there exists a pair $(\mathcal{S}, \mathcal{L})$, where $\mathcal{S} = s_1,\ldots, s_{b'}$ ($b' \geq b$) is an ordering of a subset of $V$, $\mathcal{L} = \ell_1, \ldots, \ell_{b'}$ is a sequence of labels, and vertex $s_i \in \mathcal{S}$ is assigned label $\ell_i \in [k]$, for $1 \leq i \leq b'$, such that $s_i$ satisfies the following two constraints:
\begin{itemize}
    \item $s_i$ is a sink in $A_{\ell_i}[s_i, \ldots, s_{b'}]$, and
    \item $s_i$ is a source in $B_{\ell_i}[s_1, \ldots, s_i]$.
\end{itemize}
\end{tcolorbox}

We refer to a related problem, the \textsc{Multicolored Independent Set} problem:

\begin{tcolorbox}[colback=white, colframe=black,  
                  arc=4mm, boxrule=0.3mm, 
                  fonttitle=\bfseries]
\textsc{Multicolored Independent Set}
\\
\textbf{Input:} A graph $G$, a natural number $b$, a partitioning of $V(G)$ into $b$ sets $V_1, \ldots, V_b$.
\\
\textbf{Output:} Whether there exists a set $X$ containing exactly one vertex from each set $V_i$, such that $X$ induces an independent set in $G$. 
\end{tcolorbox}

We can think of the $k$-\textsc{Compatible Bounded Ordering} problem as a generalization of the \textsc{Multicolored Independent Set} problem, as adding a digon between each pair of vertices in a set $V' \subseteq V$ in a graph $A_\ell$ (or $B_\ell$) prohibits the inclusion of more than one vertex in $V'$ in the ordering $\mathcal{S}$ with label $\ell$. This can be used to encode the constraint concerning the selection of a single vertex per partition in \textsc{Multicolored Independent Set}.

Since the number of desired vertices in the output is now part of the input of $k$-\textsc{Compatible Bounded Ordering}, it is worth considering whether $k$-\textsc{Compatible Bounded Ordering} can be solved in time $f(b) \cdot p(n)$, where $f$ is some computable function, $p$ is a polynomial function and $n$ is the number of vertices in the input. The next result indicates that this is unlikely to be the case. 

\thmBoundedOrdering*

\begin{proof}
We prove that the statement holds even when $k = 1$. We reduce from the \textsc{Multicolored Independent Set} problem. \textsc{Multicolored Independent Set} is known to be $\mathsf{W[1]}$-hard parameterized by the size of the solution~\cite{cygan:parameterized}.

We now describe the reduction. Given an instance $(G = (V_1 \cup \ldots \cup V_k, E), b)$ of \textsc{Multicolored Independent Set}, each vertex in $V(G)$ is a vertex in the constructed vertex set $V$. For each pair of vertices $v_i, v_j$ that belong to the same partition in $G$, we add the arcs $v_iv_j$ and $v_jv_i$ to $A_1$. For each pair of vertices $v_i, v_j$ that do not belong in the same partition in $G$, such that $v_iv_j \in E$,  we add the arcs $v_iv_j$ and $v_jv_i$ to $A_1$. All of the arcs in the construction are therefore contained within $A_1$, and the graph $B_1$ does not contain any arc.

We claim that $(G, b)$ is a yes-instance of \textsc{Multicolored Independent Set} if and only if $(V, \{(A_1, B_1)\}, b)$ is a yes-instance of $1$-\textsc{Compatible Bounded Ordering}.

\noindent 
\\($\Rightarrow$) Suppose $(G, b)$ is a yes-instance of \textsc{Multicolored Independent Set}, and let $v_1, \ldots, v_b$ be the vertices that form the multicolored independent set, such that $v_i \in V_i$ for $1 \leq i \leq b$. We set $\mathcal{S}$ to be equal to $v_1, \ldots, v_{b}$, and these vertices are trivially assigned label 1. Since none of the vertices in $\mathcal{S}$ are pairwise in the same partition, and none of these vertices are pairwise adjacent in $G$, they induce an independent set in $A_1$; the ordering $\mathcal{S}$ is therefore valid with the labeling that assigns label 1 to all the vertices of $\mathcal{S}$.

\noindent 
\\($\Leftarrow$) Suppose $(S, \{(A_1 \cup B_1)\}, b)$ is a yes-instance of $1$-\textsc{Compatible Bounded Ordering}, and let $\mathcal{S} = v_1, \ldots, v_{b'}$ ($b' \geq b$) be a ordering that is valid with the labeling that assigns label 1 to every vertex. We remark that the vertices in $\mathcal{S}$ pairwise do not belong to the same partition in $G$, and are pairwise not adjacent in $G$, otherwise, they would have induced a digon in $A_1$, prohibiting the inclusion of both of them in $\mathcal{S}$, regardless of their relative positions. As such, this implies that $b' = b$, since the vertex set of $G$ is partitioned into $b$ sets. The vertices in $\mathcal{S}$ therefore form an independent set in $G$, such that each vertex belongs to a distinct partition, which makes $\mathcal{S}$ a multicolored independent set in $G$, as desired. 
\end{proof}

\section{Finding $k$-compatible set arrangements}
\label{sec:compat_set_arrangement}

While we are not able to settle the hardness of \textsc{MM-RAMP}, in a bid to better our understanding of the problem, this section's main result is a hardness result on a restricted set of $k$-\textsc{Compatible Set Arrangement} instances, which are structurally similar to those representing \textsc{MM-RAMP} instances. 

\vspace{1em}
We recall the $k$-\textsc{Compatible Set Arrangement} problem:

\begin{tcolorbox}[colback=white, colframe=black,  
                  arc=4mm, boxrule=0.3mm, 
                  fonttitle=\bfseries] 
$k$-\textsc{Compatible Set Arrangement}
\\\
\textbf{Input: } A vertex set $V$, a collection $\mathcal{G}$ of $k$ pairs of directed graphs $(A_1, B_1), \ldots, (A_k, B_k)$, such that $V(A_i) = V(B_i) = V$ for $1 \leq i \leq k$, and two sets $V_{\text{s}}, V_{\text{t}} \subseteq V$.
\\\textbf{Output: } Whether there exists a sequence of sets $\mathbb{V} = V_1, \ldots, V_q$, such that:
\begin{itemize}
    \item $V_{\text{s}} = V_1$ and $V_{\text{t}} = V_q$;
    \item $|V_j \Delta V_{j + 1}| = 1$ for $1 \leq j \leq q - 1$ (i.e., a vertex is either added to or removed from $V_j$ to obtain $V_{j + 1}$); and
    \item For each vertex $v$ that is added to or removed from a set $V_j$, there exists a label $\ell \in [k]$ assigned to $v$ such that:
    \begin{itemize}
        \item If $v$ is added to $V_j$, $v$ is a sink in $A_\ell[V_{j + 1}]$ and a source in $B_\ell[V \setminus \left(V_{j + 1} \setminus \{v\}\right)]$. 
        \item If $v$ is removed from $V_j$, $v$ is a sink in $A_\ell[V_{j}]$ and a source in $B_\ell[V \setminus \left(V_{j} \setminus \{v\}\right)]$.
    \end{itemize}
\end{itemize}  
\end{tcolorbox}

We define a \emph{legal} addition/removal to be an addition/removal of a vertex from a set, such that the addition/removal does not violate the vertex's sink and source constraints, where the notions of sink constraint and source constraint satisfaction carry over from the $k$-\textsc{Compatible Ordering} problem. 

The reason why we consider the $k$-\textsc{Compatible Set Arrangement} problem is because the \textsc{MM-RAMP} problem where each robotic arm can take on one of two angles $a_1, a_2$ (which need not be the same two angles for all robotic arms) is a special case of the $2$-\textsc{Compatible Set Arrangement} problem. Indeed, the set $V$ represents the robotic arms, the sets within the sequence $\mathbb{V}$ represent which robotic arms have angle $a_2$ at any given point, and the arcs in $(A_1, B_1)$ and $(A_2, B_2)$ encode the restrictions on the clockwise and counterclockwise rotation of the robotic arms respectively. An arc in $A_1$ from a vertex $v_i$ representing robotic arm $r_i$ to a vertex $v_j$ representing robotic arm $r_j$ signifies that the rotation of $r_i$ would hit $r_j$ if $r_j$ has angle $a_2$ and $r_i$ were to be rotated clockwise from $a_1$ to $a_2$, or counterclockwise from $a_2$ to $a_1$. An arc in $B_1$ from $v_i$ to $v_j$ signifies that the rotation of $r_j$ would hit $r_i$ if $r_i$ has angle $a_1$ and $r_j$ were to be rotated clockwise from $a_1$ to $a_2$, or counterclockwise from $a_2$ to $a_1$. The interpretation of the arcs in $(A_2, B_2)$ is identical, the only difference being the direction in which the robotic arms are being rotated. 

This section's main result is obtained via a reduction from a well-studied problem in the realm of reconfiguration problems: the \textsc{Configuration-To-Configuration Nondeterministic Constraint Logic} (\textsc{C2C NCL}) problem.

We define terminology used in the definition of the \textsc{C2C NCL} problem:

\begin{definition}
    An undirected weighted cubic graph $G = (V, E)$ where each vertex is either incident to three edges with weight 2 (an \textbf{OR vertex}), or one edge with weight 2 and two edges with weight 1 (an \textbf{AND vertex}), is called a \textbf{constraint graph}.
\end{definition}
\begin{definition}
    A set of edge orientations $E^o$ of edges in a constraint graph $G = (V, E)$ is said to be a \textbf{legal set of edge orientations} if and only if each vertex has an in-weight of at least 2.
\end{definition}

We can now formally define the \textsc{C2C NCL} problem:

\begin{tcolorbox}[colback=white, colframe=black,  
                  arc=4mm, boxrule=0.3mm, 
                  fonttitle=\bfseries] 
$k$-\textsc{Configuration-To-Configuration Nondeterministic Constraint Logic (C2C NCL)}
\\\
\textbf{Input: } A constraint graph $G = (V, E)$, two legal sets of edge orientations of $E$, $E^o_{\text{s}}$ and $E^o_{\text{t}}$.
\\\textbf{Output: } Whether there exists a sequence of legal edge orientations sets $\mathbb{E} = E^o_1, \ldots, E^o_{q}$, such that:
\begin{itemize}
    \item $E^o_{\text{s}} = E^o_1$ and $E^o_{\text{t}} = E^o_q$; and
    \item $|E^o_j \Delta E^o_{j + 1}| = 1$ for $1 \leq j \leq q - 1$ (i.e., an edge has its orientation reversed to obtain $E^o_{j+1}$ from $E^o_j$).
\end{itemize}  
\end{tcolorbox}

Starting from a legal set of edge orientations, we provide intuition for when an edge orientation reversal yields another legal set of edge orientations. 

When it comes to OR vertices, the orientation of an edge can be changed from inward to outward if and only if one of the other two edges incident to the same OR vertex is oriented inward, so there is a ``choice'' for which of the other two edges to ``rely on'' to guarantee that the in-weight of the vertex remains at least 2 after the edge orientation change; this choice will be encoded in our construction by the choice of a label for a vertex addition/removal, where the vertex addition/removal is associated with the edge orientation change. 

As for AND vertices, the orientation of a weight 2 edge can be changed from inward to outward if and only if both the other edges incident to the same AND vertex are oriented inward, and the orientation of a weight 1 edge can be changed from inward to outward if and only if the weight 2 edge incident to the same AND vertex is oriented inward. As such, the reversal of the orientation of an edge in an AND vertex can take place if and only if a unique condition is met, regardless of the weight of the edge: this is not the case for edge orientation changes in OR vertices. In line with the intuition for the labels of vertex additions/removals associated with edge orientation changes in OR vertices, the label for a vertex addition/removal associated with an edge orientation change in an AND vertex will be irrelevant in our construction.

In view of the $k$-\textsc{Compatible Ordering} problem and the $k$-\textsc{Compatible Set Arrangement} problem being generalizations of \textsc{SM-RAMP} and \textsc{MM-RAMP} respectively, as well as the results presented so far, we are mostly interested in instances of $k$-\textsc{Compatible Set Arrangement} that represent \textsc{MM-RAMP} instances, i.e., that are characterized by planarity and low degree. Our next result shows that the \textsc{C2C NCL} problem can be used to prove that $k$-\textsc{Compatible Set Arrangement} problem is $\mathsf{PSPACE}$-complete even when restricted to such instances.

\thmArrangement*

\begin{proof}
We prove that $k$-\textsc{Compatible Set Arrangement} is $\mathsf{PSPACE}$-complete on instances where the union is planar, has maximum degree 6, and has bounded bandwidth even when $k = 2$. We reduce from \textsc{C2C NCL}, which is known to be $\mathsf{PSPACE}$-complete even when restricted to instances where the constraint graph is planar and has bounded bandwidth~\cite{vandz:parameterized}.

We describe how an arbitrary instance $(G = (V, E), E^o_{\text{s}}, E^o_{\text{t}})$ of \textsc{C2C NCL} can be transformed into an equivalent instance of \textsc{2-Compatible Set Arrangement}, starting with the construction of the vertex set.

For each edge $e = v_iv_j \in E$, we include two vertices in the vertex set $V$ of the constructed \textsc{2-Compatible Set Arrangement} instance, which we call \emph{endpoint vertices}, and which we denote by $v_i^e$ and $v_j^e$ respectively. The endpoint vertices $v_i^e, v_j^e$ are said to be \emph{associated with} the edge $e$. We intend to use the membership of an endpoint $v_i^e$ in a set $V_j \in \mathbb{V}$ to indicate that the endpoint $v_i$ is currently the head of the oriented version of $e$. In particular, this implies that the construction should not allow both $v_i^e$ and $v_j^e$ to be included in any of the sets in $\mathbb{V}$.

The intended interpretation of the endpoint vertices allows us to describe the initial and the final vertex sets: if $v_i^e$ (resp. $v_j^e$) is the head of the oriented edge $e$ in $E^o_{\text{s}}$,  $v_i^e$ (resp. $v_j^e$) is added to $V_{\text{s}}$. $V_{\text{t}}$ is constructed similarly out of $E^o_{\text{t}}$. In particular, we have: $|E^o_{\text{s}}| = |E^o_{\text{t}}| = |V_{\text{s}}| = |V_{\text{t}}| = |E|$.

Next, we cover the construction of the two types of arcs: arcs encoding edges in $G$, and arcs encoding the relationship between edges incident to the same vertex in $G$. The latter of those two arc types is handled across two cases depending on the type of the vertex they concern:

\begin{enumerate}
    \item \textbf{Arcs encoding edges in G: } For each pair of endpoint vertices $v_i^e, v_j^e$ associated with the same edge $e = v_iv_j$ in $G$, we add the two arcs  $v_i^ev_j^e$ and $v_j^ev_i^e$ to both $A$ graphs. As stated earlier, the purpose of these digons is to ensure that there does not exist a set in $\mathbb{V}$ that contains both $v_i^e$ and $v_j^e$.
    \item \textbf{Arcs encoding the relationship between edges incident to the same OR vertex:} For each OR vertex $v_i$, incident to three edges $e_1, e_2, e_3$, we present the neighbors of $v_i^{e_3}$ first: we add the arcs $v^{e_1}_iv^{e_3}_i$ and $v^{e_2}_iv^{e_3}_i$ to $B_1$ and $B_2$ respectively. Consider an arbitrary vertex set $V_j$, such that $v_i^{e_3}$ is the only vertex in $\{v_i^{e_1}, v_i^{e_2}, v_i^{e_3}\}$ that is in $V_j$. The removal of $v_i^{e_3}$ from $V_j$ is not possible, because in each of $B_1$ and $B_2$, $v_i^{e_3}$ has an in-neighbor that is not in $V_j$, i.e., the removal would violate its source constraint. The neighbors of $v_i^{e_2}$ and $v_i^{e_3}$ mirror those of $v_i^{e_3}$: we add the arcs $v_i^{e_2}v_i^{e_1}, v_i^{e_1}v_i^{e_2}$ to $B_1$ and $v_i^{e_3}v_i^{e_1}, v_i^{e_3}v_i^{e_2}$ to $B_2$.
    \item \textbf{Arcs encoding the relationship between edges incident to the same AND vertex:} For each AND vertex $v_i$, incident to three edges $e_1, e_2, e_3$, we assume, without loss of generality, that $e_1$ is the weight 2 edge. We will handle the construction for the weight 1 edges and the weight 2 edges separately. For the weight 1 edge $e_2$ (resp. $e_3$), we add the arc $v_i^{e_1}v_i^{e_2}$ (resp. $v_i^{e_1}v_i^{e_3}$) to both $B$ graphs. This prohibits the removal of $v_i^{e_2}$ (resp. $v_i^{e_3}$) from an arbitrary vertex set $V_j$ unless $v_i^{e_1} \in V_j$; otherwise, $v_i^{e_2}$'s (resp. $v_i^{e_3}$'s) source constraint is violated, regardless of its assigned label. As for the weight 2 edge $e_1$, we add arcs $v_i^{e_2}v_i^{e_1}$ and $v_i^{e_3}v_i^{e_1}$ to both $B$ graphs. This prohibits the removal of $v_{i}^{e_1}$ from an arbitrary vertex set $V_j$ unless both $v_i^{e_2}$ and $v_i^{e_3}$ are in $V_j$; otherwise, $v_i^{e_1}$'s source constraint is violated, regardless of its assigned label.
\end{enumerate}

We remark that the AND vertices arcs, unlike the OR vertices arcs, are symmetric across both $B$ graphs: this aligns with the intuition presented earlier regarding the presence of a choice for OR vertices, and the absence of a choice for AND vertices. We highlight that the union of the constructed digraphs is isomorphic to $G$, up to replacing vertices with digons and undirected edges with arcs or digons. This implies that the union of $A_1, B_1, A_2$ and $B_2$ is planar, and has bounded maximum degree: this can be seen in Figure~\ref{fig:ncl_reduction}.

\begin{figure}
    \centering
    \includegraphics[width=\linewidth]{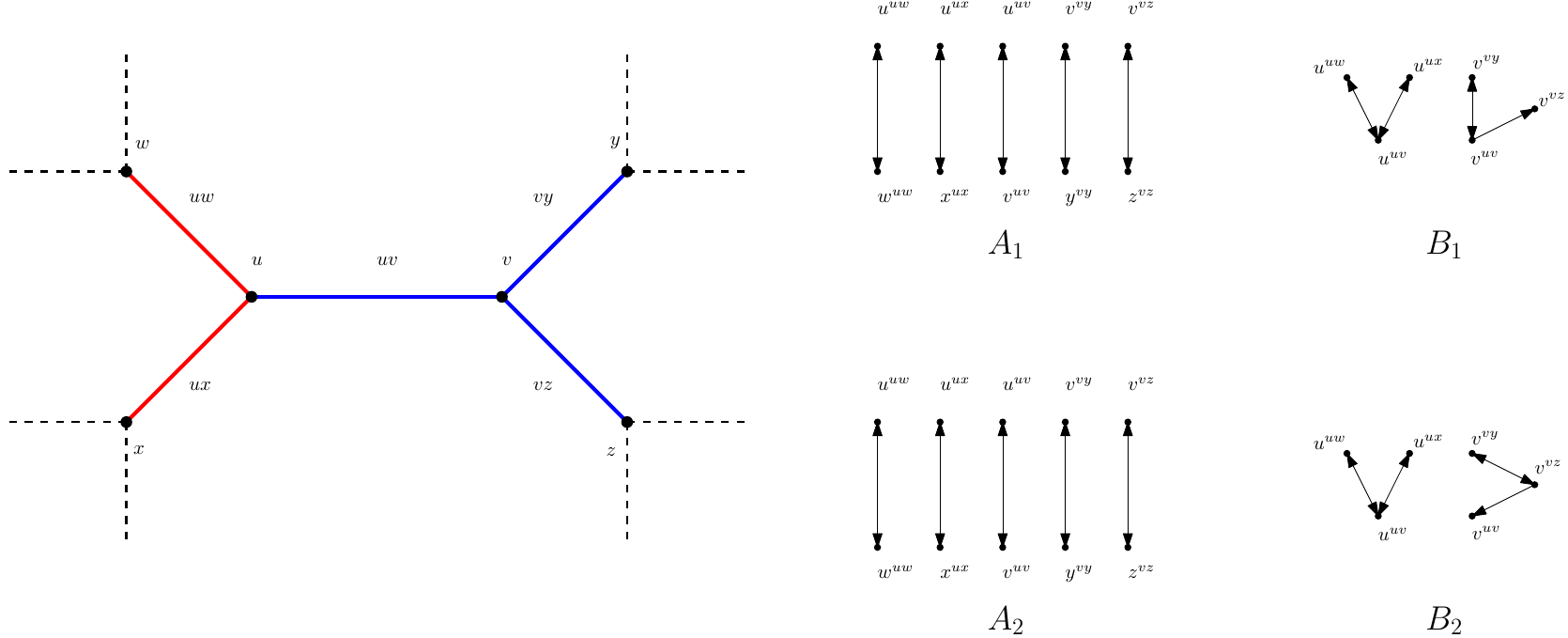}
    \caption{Illustration of the reduction from Theorem~\ref{thm:arrangement}. Only the vertices and the arcs constructed out of $u$, $v$, the edges incident to them and their neighbors are included in the illustration.}
    \label{fig:ncl_reduction}
\end{figure}

We claim that $(G, E^o_{\text{s}}, E^o_{\text{t}})$ is a yes-instance of \textsc{C2C NCL} if and only if $(V, \mathcal{G} = \{(A_1, B_1), (A_2, B_2)\}, V_{\text{s}}, V_{\text{t}})$ is a yes-instance of $k$-\textsc{Compatible Set Arrangement}. 

Before delving into the details of the proof, it will be useful to introduce two observations that greatly restrict the structure of a solution to the constructed $k$-\textsc{Compatible Set Arrangement} instance:

\begin{observation}
    \label{obs:sinkremove}
    Given a vertex set $V_i$ containing a single vertex per digon, removing a vertex $v \in V_i$ does not violate vertex $v$'s sink constraint.  
\end{observation}

\begin{proof}
    Towards a contradiction, assume otherwise. Then, there must exist a vertex $w$ in $V_i$ whose removal violates vertex $v$'s sink constraint, which implies that $vw$ is an arc in both $A$ graphs. However, if $vw$ is an arc, by construction, so is $wv$, which contradicts the fact that $V_i$ contains a single vertex per digon.
\end{proof}

\begin{observation}
    \label{obs:sinkadd}
    Given a vertex set $V_i$ containing a single vertex per digon, adding a vertex $v$ to $V_i$ violates vertex $v$'s sink constraint.  
\end{observation}

\noindent The proof of Observation~\ref{obs:sinkadd} is identical to that of Observation~\ref{obs:sinkremove} because vertex additions and vertex removals are treated identically as far as constraint satisfaction is concerned, i.e., the removal of a vertex $v$ from $V_i$ is legal if and only if the addition of a vertex $v$ to $V_i \setminus \{v\}$ is legal.
\\
\\
($\Rightarrow$) Let $\mathbb{E} = (E^o_{\text{s}} = E^o_1, \ldots, E^o_{q} = E^o_{\text{t}})$ be a sequence of edge orientations that is a solution to the \textsc{C2C NCL} instance (i.e., a single edge's orientation is reversed between $E^o_i$ and $E^o_{i + 1}$ for $1 \leq i \leq q-1$, and all the edge orientations in $\mathbb{E}$ are legal). Starting from $\mathbb{E}$, we construct a sequence of vertex sets $\mathbb{V} = (V_{\text{s}} = V_1, \ldots, V_{2q} = V_{\text{t}})$ that is a solution to the constructed $2$-\textsc{Compatible Set Arrangement} instance. The argument is inductive; in the sequence of vertex sets we construct, we prove that all vertex additions and removals are legal, and that the vertices in $V_{2i}$ correspond to the heads of the directed edges in $E^o_i$ for $1 \leq i \leq q$.

The statement we are trying to prove clearly holds for $E^o_1 = E^o_{\text{s}}$ and $V_1 = V_{\text{s}}$, since by construction the vertices in $V_{\text{s}}$ correspond to the heads of the oriented edges in $E_{\text{s}}^o$.

By the inductive hypothesis, we assume that, up until an arbitrary index $c$, $1 \leq c \leq q - 1$, all the vertex additions and removals that were executed are legal, and that the vertices in $V_{2c}$ correspond to the heads of the oriented edges in $E^o_c$. Now, we consider an edge orientation change going from the set $E^o_c$ to the set $E^o_{c + 1}$. In contrast to earlier mentions of edges in the paper, the names of the edges in the discussion to follow are not indicative of their orientation; their orientation, if of interest, will be clearly specified. We refer to the reversed (undirected) edge as $uv$. Before the flip, $uv$ is directed towards $u$, a vertex that is incident to (undirected) edges $uw$ and $ux$, and away from $v$, a vertex that is incident to two other (undirected) edges $vy$ and $vz$.  

We present a case analysis on the vertex type of $u$ and that of $v$ to prove the existence of a legal vertex removal, followed by a legal vertex addition, that can be applied to $V_{2c}$ to yield a vertex set $V_{2c + 2}$ whose vertices are the heads of the oriented edges of $E^o_{c + 1}$. It is sufficient to check whether the first of the two operations, i.e., the vertex removal, satisfies the constraints of the removed vertex, as adding a vertex is a symmetrical operation, and is covered by a proof of correctness that replicates that for the vertex removal. 

Analyzing whether the vertex removal we propose is legal will omit any mention of the $A$ graphs, because we have already established that a vertex removal cannot violate the removed vertex's sink constraint (Observation~\ref{obs:sinkremove}):

\begin{itemize}
\item \textbf{The vertex $u$ is an OR vertex:} At least one of $uw$ and $ux$ must be oriented towards $u$; otherwise reversing the orientation of $uv$ would not yield a legal edge orientation. Without loss of generality, assume that $uw$ is oriented towards $u$, i.e., $u^{uw}$ is in $V_{2c}$ by the inductive hypothesis. The construction guarantees that there exists a label $j$ such that $u^{uw}u^{uv}$ is an arc in $B_j$, and such that $u^{uv}$ has no other in-neighbors in $B_j$. As such, $V_{2c + 1} = V_{2c} \setminus \{u^{uv}\}$ and $u^{uv}$ is assigned label $j$.
\item \textbf{The vertex $u$ is an AND vertex, and $uv$ is a weight 2 edge:} Both $uw$ and $ux$ must be oriented towards $u$, i.e., $u^{uw}$ and $u^{ux}$ are vertices in $V_{2c}$ by the inductive hypothesis. The construction guarantees that $u^{uv}$ has two in-neighbors in both $B_1$ and $B_2$: $u^{uw}$ and $u^{ux}$, both of which are in $V_{2c}$, and no other in-neighbors. As such, $V_{2c + 1} = V_{2c} \setminus \{u^{uv}\}$ and $u^{uv}$ is assigned label 1 (without loss of generality, since the construction across the two $B$ graphs is symmetrical).
\item \textbf{The vertex $u$ is an AND vertex, and $uv$ is a weight 1 edge:} Without loss of generality, assume that $uw$ is a weight 2 arc: $uw$ must be oriented towards $u$, i.e., $u^{uw}$ is a vertex in $V_{2c}$ by the inductive hypothesis. The construction guarantees that $u^{uv}$ has $u^{uw}$ as an in-neighbor in both $B_1$ and $B_2$, and no other in-neighbors. As such, $V_{2c + 1} = V_{2c} \setminus \{u^{uv}\}$ and $u^{uv}$ is assigned label 1 (without loss of generality, since the construction across the two $B$ graphs is symmetrical).
\end{itemize}

Next, we seek to add $v^{uv}$ to $V_{2c + 1}$. As stated earlier, the structure of the proof for the legality of this vertex addition is identical to the proof for the legality of the removal of $u^{uv}$ from $V_{2c}$.

Therefore, $V_{2c + 2} = (V_{2c} \setminus \{u^{uv}\}) \cup \{v^{uv}\}$, and both the removal of $u^{uv}$ and the addition of $v^{uv}$ are legal. Starting with two sets $V_{2c}$ and $E^o_c$, such that $V_{2c}$ contains the heads of the oriented edges of $E^o_c$, by the inductive hypothesis, we were able to construct a set $V_{2c + 2}$ that contains the heads of the oriented edges of $E^o_{c + 1}$, after applying a legal vertex removal, followed by a legal vertex addition.

Since $c$ is arbitrary, it follows that $V_{2\text{t}}$ contains all the heads of the directed arcs in $E_{\text{t}}$, as desired. This concludes the forward direction of the proof.

\noindent
\\($\Leftarrow$) Let $\mathbb{V} = (V_{\text{s}} = V_1, \ldots, V_q = V_{\text{t}})$ be a sequence of vertex sets that is a solution to the constructed $2$-\textsc{Compatible Set Arrangement} instance. We reuse parts of the argument for the forward direction of the proof. In particular, the forward direction maps each edge orientation to a vertex addition followed by a vertex removal that act on the two endpoint vertices of the same reversed edge. It is therefore sufficient to prove that the addition/removal sequence associated with $\mathbb{V}$ can be partitioned into pairs, where each pair is a removal of an endpoint vertex and an addition of the other endpoint vertex of the same edge, to be able to do the inverse of this mapping.

We introduce a definition and an observation that will be fundamental in our proof arguments, and that can directly be derived from the definition of the problem. We define vertex addition applied to a sequence of vertex sets: adding a vertex $v$ to a vertex set sequence $V_1, \ldots V_j$, where $v$ is not contained in any of the vertex sets in the sequence, yields the vertex set sequence $V_1, V_1 \cup \{v\}, V_2 \cup \{v\}, \ldots, V_{j-1} \cup \{ v \}, V_j \cup \{ v \}, V_j$; removing a vertex from a vertex set sequence, where $v$ is contained in all the vertex sets in the sequence, is defined analogously.

\begin{observation}
    \label{obs:addremove}
    Let $\mathbb{V} = (V_1, \ldots, V_i, \ldots, V_j, \ldots V_q)$ be a solution to a $k$-\textsc{Compatible Set Arrangement} instance. The addition (resp. the removal) of a vertex $v$  to (resp. from) the sequence $V_i, \ldots, V_j$ cannot violate the (previously satisfied) source (resp. sink) constraints of any of the vertices being added or removed in $V_i, \ldots, V_j$, and can only violate their (previously satisfied) sink (resp. source) constraints. 
\end{observation}

We can assume that the sequence $\mathbb{V}$ has a structure that is implied by the following observations:

\begin{observation}
    \label{obs:oneremoval}
    There is at least one vertex removal operation in $\mathbb{V}$, unless $V_{\text{s}} = V_{\text{t}}$.
\end{observation}
\begin{proof}
    If $V_{\text{s}} = V_{\text{t}}$, there is nothing to do. Otherwise, we can refer to Observation~\ref{obs:sinkadd} to conclude the existence of at least one vertex removal in $\mathbb{V}$. 
\end{proof}

\begin{observation}
    \label{obs5}
    Let $e = uv$ such that $v^e$ is removed from $V_{j_1} \in \mathbb{V}$. In a shortest solution $\mathbb{V}$ we may assume there exists an index $j_2$ such that $u^e$  is added to $V_{j_2} \in \mathbb{V}$ where $j_1 < j_2$. 
\end{observation}
\begin{proof}
    Suppose otherwise, for the sake of contradiction. If the removal of $v^e$ is the last addition/removal operation involving any of $v^e$ and $u^e$ in $\mathbb{V}$, we can find a shorter sequence of vertex sets that is a solution to the instance. If $\mathbb{V} = V_1, \ldots, V_{j_1}, V_{j_1 + 1}, V_{j_1 + 2},  \ldots, V_q$ is a solution to the instance, so is  $\mathbb{V}' = V_1, \ldots, V_{j_1}, V_{j_1 + 2} \cup \{v^e\}, \ldots, V_q \cup \{v^e\}$. Undoing the removal of $v^e$ can only cause sink constraint violations for the vertices that are subsequently added or removed (Observation~\ref{obs:addremove}), and the only vertex whose sink constraint involves $v^e$ is $u^e$. A similar argument can be used to shorten the vertex set sequence that is a solution to the instance if the removal of $v^e$ is followed by the addition of $v^e$ to some vertex set $V_{j_3}$, where $j_3 > j_1$.
\end{proof}

In what follows, we assume that $\mathbb{V}$ is a nontrivial shortest solution whose structure adheres to that of the solution in the statement of Observation~\ref{obs5}.

\begin{observation}
    \label{obs6}
     Let $e$ be an edge with endpoints $u$ and $v$, and let $j_1$ and $j_2$ ($j_1 < j_2$) be two indices such that $u^e$ is removed from $V_{j_1}$,  $v^e$ is added to $V_{j_2}$, and neither $u^e$ nor $v^e$ is added to or removed from $V_{j_1 + 1}, \ldots, V_{j_2 - 1}$ in $\mathbb{V}$. There exists a shortest solution $\mathbb{V}'$ where $j_2 = j_1 + 1$.
\end{observation}
\begin{proof}
    We modify $\mathbb{V}$ to obtain $\mathbb{V}'$ by moving $v^e$ directly after the removal operation of $u^e$, which can be thought of as the undoing of a vertex addition of $v^e$, followed by a vertex addition of $v^e$, both to a sequence of vertex sets. This modification can only violate sink constraints (Observation~\ref{obs:addremove}), and the only vertex that is potentially affected by the earlier inclusion of $v^e$ is $u^e$. However, given our choice of $j_1$ and $j_2$, none of the operations that happen on vertex sets between $V_{j_1}$ and $V_{j_2}$ involve $u^e$ or $v^e$, so if the operations taking place between $V_{j_1}$ and $V_{j_2}$ are legal in $\mathbb{V}$, they will remain legal in $\mathbb{V}'$.
\end{proof}

Since we know that if we are dealing with a nontrivial sequence, there exists a solution $\mathbb{V}$ with at least one endpoint vertex removal (Observation~\ref{obs:oneremoval}), such that each endpoint vertex removal $u^{uv}$ is directly followed by an endpoint vertex addition $v^{uv}$ (Observations~\ref{obs5} and~\ref{obs6}). Moreover, if $i$ is odd, vertex set $V_{i + 1}$ in $\mathbb{V}$ is obtained from $V_i$ by removing an endpoint vertex $u^{uv}$, and vertex set $V_{i + 2}$ is obtained from $V_{i + 1}$ by adding the other endpoint vertex $v^{uv}$ associated with the same edge. It follows that that $|V_i| = |E|$ if $i$ is odd and $|V_{i + 1}| = |E| - 1$ if $i$ is even.

Now, starting from $\mathbb{V}$, we are able to construct a sequence of edge orientation sets $\mathbb{E} = (E^o_{\text{s}} = E^o_1, \ldots, E^o_{q/2} = E^o_{\text{t}})$ that is a solution to the constructed $2$-\textsc{Compatible Set Arrangement} instance by executing the inverse of the mapping that was presented in the forward direction. 
\end{proof}

\section{Conclusion and open problems}\label{sec:conclusion}
In this paper, we introduced two geometric reconfiguration problems, \textsc{SM-RAMP} and \textsc{MM-RAMP}, as well as two graph problems which generalize them, the $k$-\textsc{Compatible Ordering} problem and the $k$-\textsc{Compatible Set Arrangement Problem}, and we settled the complexity of the latter two problems under myriad constraints, namely acyclicity, planarity, bounded maximum degree, and degeneracy, towards a better understanding of the two initial problems on robotic arms. We also discussed two graph parameters, the treewidth and the labeled modular width, which we believe are bounded in $k$-\textsc{Compatible Ordering} instances that characterize \textsc{SM-RAMP} instances, and we designed two algorithms that efficiently solve the $k$-\textsc{Compatible Ordering} problem on instances where either of those two parameters are bounded, which further extends the set of \textsc{SM-RAMP} instances that we know how to efficiently solve.

We propose several open problems. First, although we have shown that several natural generalizations of \textsc{SM-RAMP} are $\mathsf{NP}$-complete, the complexity of \textsc{SM-RAMP} itself remains unresolved. It would be valuable to determine whether this specific problem admits a polynomial-time solution or if it too is $\mathsf{NP}$-complete. 
One might also explore what happens when the problem is relaxed or generalized further. For instance, what is the complexity of the problem if we remove the assumption that all robotic arms have the same length? The introduction of broader constraints could significantly impact the problem's tractability. 
Moreover, as far as our positive results for $k$-\textsc{Compatible Ordering} are concerned, knowing whether a labeled modular decomposition can be computed efficiently makes our current result on instances of bounded labeled modular width all the more robust, and may ameliorate our understanding of the problems on robotic arms.

Another intriguing direction concerns the complexity of $k$-\textsc{Compatible Ordering} under geometric constraints, as opposed to purely topological ones like planarity. For example, what is the complexity of $k$-\textsc{Compatible Ordering} when the graphs in $\mathcal{G}$ are unit disk graphs, which arise naturally in scenarios involving spatial or wireless network constraints? These graphs exhibit geometric properties that could either simplify or complicate the problem.

Finally, while we explored certain structural graph properties such as bounded treewidth and labeled modular width, the algorithmic impact of other structural graph parameters such as clique-width or bounded expansion on $k$-\textsc{Compatible Ordering} remains uncharted. Addressing these questions would deepen our understanding of the interplay between graph structure and reconfiguration complexity.

\section*{Acknowledgments}

Nicolas Bousquet was partly supported by ANR project ENEDISC (ANR-24-CE48-7768-01). Remy El Sabeh and Naomi Nishimura's research was funded by the Natural Science and Engineering Research Council of Canada.

\bibliographystyle{abbrv} 
\bibliography{main}

\begin{thebibliography}{10}

\bibitem{bodlaender:treewidth}
H.~L. Bodlaender and T.~Kloks.
\newblock Efficient and constructive algorithms for the pathwidth and treewidth of graphs.
\newblock {\em Journal of Algorithms}, 21(2):358--402, 1996.

\bibitem{bousquet:feedback}
N.~Bousquet, F.~Hommelsheim, Y.~Kobayashi, M.~M{\"u}hlenthaler, and A.~Suzuki.
\newblock Feedback vertex set reconfiguration in planar graphs.
\newblock {\em Theoretical Computer Science}, 979:114188, 2023.

\bibitem{BousquetMNS22+}
N.~Bousquet, A.~E. Mouawad, N.~Nishimura, and S.~Siebertz.
\newblock A survey on the parameterized complexity of reconfiguration problems.
\newblock {\em Comput. Sci. Rev.}, 53:100663, 2024.

\bibitem{brianski:interval}
M.~Brianski, S.~Felsner, J.~Hodor, and P.~Micek.
\newblock Reconfiguring independent sets on interval graphs.
\newblock In F.~Bonchi and S.~J. Puglisi, editors, {\em 46th International Symposium on Mathematical Foundations of Computer Science, {MFCS} 2021, August 23-27, 2021, Tallinn, Estonia}, volume 202 of {\em LIPIcs}, pages 23:1--23:14. Schloss Dagstuhl - Leibniz-Zentrum f{\"{u}}r Informatik, 2021.

\bibitem{buchin:dots}
K.~Buchin, M.~Hagedoorn, I.~Kostitsyna, and M.~van Mulken.
\newblock Dots {\&} boxes is {PSPACE}-complete.
\newblock In F.~Bonchi and S.~J. Puglisi, editors, {\em 46th International Symposium on Mathematical Foundations of Computer Science, {MFCS} 2021, August 23-27, 2021, Tallinn, Estonia}, volume 202 of {\em LIPIcs}, pages 25:1--25:18. Schloss Dagstuhl - Leibniz-Zentrum f{\"{u}}r Informatik, 2021.

\bibitem{DBLP:journals/eatcs/CourcelleE12}
B.~Courcelle and J.~Engelfriet.
\newblock Book: Graph structure and monadic second-order logic. {A} language-theoretic approach.
\newblock {\em Bull. {EATCS}}, 108:179, 2012.

\bibitem{CurkovicJ10}
P.~{\'{C}}urkovi{\'{c}} and B.~Jerbi{\'{c}}.
\newblock Dual-arm robot motion planning based on cooperative coevolution.
\newblock In L.~M. Camarinha-Matos, P.~Pereira, and L.~Ribeiro, editors, {\em Emerging Trends in Technological Innovation}, pages 169--178, Berlin, Heidelberg, 2010. Springer Berlin Heidelberg.

\bibitem{cygan:parameterized}
M.~Cygan, F.~V. Fomin, {\L}.~Kowalik, D.~Lokshtanov, D.~Marx, M.~Pilipczuk, M.~Pilipczuk, and S.~Saurabh.
\newblock {\em Parameterized algorithms}, volume~5.
\newblock Springer, 2015.

\bibitem{demaine:tetris}
E.~D. Demaine, S.~Hohenberger, and D.~Liben-Nowell.
\newblock Tetris is hard, even to approximate.
\newblock In {\em Computing and Combinatorics: 9th Annual International Conference, COCOON 2003 Big Sky, MT, USA, July 25--28, 2003 Proceedings 9}, pages 351--363. Springer, 2003.

\bibitem{gajarsky:parameterized}
J.~Gajarsk{\`y}, M.~Lampis, and S.~Ordyniak.
\newblock Parameterized algorithms for modular-width.
\newblock In {\em Parameterized and Exact Computation: 8th International Symposium, IPEC 2013, Sophia Antipolis, France, September 4-6, 2013, Revised Selected Papers 8}, pages 163--176. Springer, 2013.

\bibitem{hearn:sliding}
R.~A. Hearn and E.~D. Demaine.
\newblock {PSPACE}-completeness of sliding-block puzzles and other problems through the nondeterministic constraint logic model of computation.
\newblock {\em Theoretical Computer Science}, 343(1-2):72--96, 2005.

\bibitem{hearn:games}
R.~A. Hearn and E.~D. Demaine.
\newblock {\em Games, puzzles, and computation}.
\newblock CRC Press, 2009.

\bibitem{HueA22}
L.~T. Hue and N.~P.~T. Anh.
\newblock Planning the optimal trajectory for a dual-arm robot system using a genetic algorithm considering the controller.
\newblock In {\em 2022 International Conference on Advanced Technologies for Communications (ATC)}, pages 292--297, 2022.

\bibitem{ito:reconf}
T.~Ito, E.~D. Demaine, N.~J. Harvey, C.~H. Papadimitriou, M.~Sideri, R.~Uehara, and Y.~Uno.
\newblock On the complexity of reconfiguration problems.
\newblock {\em Theoretical Computer Science}, 412(12-14):1054--1065, 2011.

\bibitem{mcconnell:moddecomp}
R.~M. McConnell and J.~P. Spinrad.
\newblock Modular decomposition and transitive orientation.
\newblock {\em Discrete Mathematics}, 201(1-3):189--241, 1999.

\bibitem{mouawad:parameterized}
A.~E. Mouawad, N.~Nishimura, V.~Raman, and M.~Wrochna.
\newblock Reconfiguration over tree decompositions.
\newblock In {\em International Symposium on Parameterized and Exact Computation}, pages 246--257. Springer, 2014.

\bibitem{Nishimura17}
N.~Nishimura.
\newblock Introduction to reconfiguration.
\newblock {\em Algorithms}, 11(4):52, 2018.

\bibitem{ohsaka:gap2022}
N.~Ohsaka.
\newblock Gap preserving reductions between reconfiguration problems.
\newblock In P.~Berenbrink, P.~Bouyer, A.~Dawar, and M.~M. Kant{\'{e}}, editors, {\em 40th International Symposium on Theoretical Aspects of Computer Science, {STACS} 2023, March 7-9, 2023, Hamburg, Germany}, volume 254 of {\em LIPIcs}, pages 49:1--49:18. Schloss Dagstuhl - Leibniz-Zentrum f{\"{u}}r Informatik, 2023.

\bibitem{ohsaka:gap2024}
N.~Ohsaka.
\newblock Gap amplification for reconfiguration problems.
\newblock In {\em Proceedings of the 2024 Annual ACM-SIAM Symposium on Discrete Algorithms (SODA)}, pages 1345--1366. SIAM, 2024.

\bibitem{SinaFB18}
S.~S.~M. Salehian, N.~Figueroa, and A.~Billard.
\newblock A unified framework for coordinated multi-arm motion planning.
\newblock {\em The International Journal of Robotics Research}, 37(10):1205--1232, 2018.

\bibitem{tedder:moddecomp}
M.~Tedder, D.~Corneil, M.~Habib, and C.~Paul.
\newblock Simpler linear-time modular decomposition via recursive factorizing permutations.
\newblock In {\em Automata, Languages and Programming: 35th International Colloquium, ICALP 2008, Reykjavik, Iceland, July 7-11, 2008, Proceedings, Part I 35}, pages 634--645. Springer, 2008.

\bibitem{tippenhauer:3sat}
S.~Tippenhauer and W.~Muzler.
\newblock On planar 3-{SAT} and its variants.
\newblock {\em Fachbereich Mathematik und Informatik der Freien Universitat Berlin}, 2016.

\bibitem{Heuvel13}
J.~van~den Heuvel.
\newblock {\em {The Complexity of change}}, page 409.
\newblock {Part of London Mathematical Society Lecture Note Series}. 2013.

\bibitem{vandz:parameterized}
T.~C. van~der Zanden.
\newblock Parameterized complexity of graph constraint logic.
\newblock In T.~Husfeldt and I.~A. Kanj, editors, {\em 10th International Symposium on Parameterized and Exact Computation, {IPEC} 2015, September 16-18, 2015, Patras, Greece}, volume~43 of {\em LIPIcs}, pages 282--293. Schloss Dagstuhl - Leibniz-Zentrum f{\"{u}}r Informatik, 2015.

\end{thebibliography}






\newpage

\end{document}